\documentclass{article}

\usepackage[english]{babel}

\usepackage[letterpaper,top=2cm,bottom=2cm,left=3cm,right=3cm,marginparwidth=1.75cm]{geometry}

\usepackage{comment}
\usepackage{amssymb}
\usepackage{caption}
 \usepackage{natbib}
\usepackage{geometry}
\usepackage{graphicx}

\usepackage{subfig}
\usepackage{caption}
\usepackage{amsmath}
\usepackage{amsthm}
\usepackage{booktabs}   
\usepackage{float}      
\usepackage{authblk}

\usepackage{subcaption}  
\newcommand{\keywords}[1]{\par\noindent\textbf{Keywords:} #1}

\usepackage{float}
\usepackage{enumerate}
\usepackage[colorlinks=true, allcolors=blue]{hyperref}
\usepackage{soul}

\author[1]{Jiangjing Zhou}
\author[2,*]{Ovanes Petrosian}
\author[2,3]{Ye Zhang}
\author[1,*]{Hongwei Gao}

\affil[1]{School of Mathematics and Statistics, Qingdao University, 
No. 308 Ningxia Road, Qingdao, 266071, Shandong Province, China. {zhoujiangjingqdu@163.com};{gaohongwei@qdu.edu.cn}}

\affil[2]{MSU-BIT-SMBU Joint Research Center of Applied Mathematics, 
Shenzhen MSU-BIT University, 
1 International University Park Road, Dayun New Town, Longgang District, Shenzhen, 518115, Guangdong Province, China. 
{petrosian.ovanes@yandex.ru}; {ye.zhang@smbu.edu.cn}}

\affil[3]{School of Mathematics and Statistics, 
Beijing Institute of Technology, 
5 Zhongguancun South Street, Haidian District, Beijing, 100081, China.}

\title{Nash Equilibrium and Belief Evolution in Differential Games}

\date{}

\newtheorem{theorem}{Theorem}
\newtheorem{proposition}{Proposition}%
\newtheorem{definition}{Definition}%
\begin{document}
\maketitle
\begin{abstract}
This study investigates differential games with motion-payoff uncertainty in continuous-time settings. We propose a framework where players update their beliefs about uncertain parameters using continuous Bayesian updating. Theoretical proofs leveraging key probability theorems demonstrate that players' beliefs converge to the true parameter values, ensuring stability and accuracy in long-term estimations. We further derive Nash Equilibrium strategies with continuous Bayesian updating for players, emphasizing the role of belief updates in decision-making processes. Additionally, we establish the convergence of Nash Equilibrium strategies with continuous Bayesian updating. The efficacy of both continuous and dynamic Bayesian updating is examined in the context of pollution control games, showing convergence in players' estimates under small time intervals in discrete scenarios.
\end{abstract}

\section{Introduction}
Differential games \cite{dockner2000differential,friedman2013differential} involve multiple players controlling a dynamical system through their actions, which are described by differential state equations. These games evolve over a continuous-time horizon, where each player seeks to optimize an objective function that depends on the system’s state, their own actions, and potentially the actions of others. In this study, we extend the classic differential game model to scenarios involving motion-payoff uncertainty, where players face uncertainties in both the dynamic equations and the payoff functions, and are unaware of certain parameters in the environment or in their opponents’ payoff structures.

In dynamic games, optimal control techniques are generalized to accommodate multiple players with both shared and conflicting interests. As shown in \cite{lewis2012optimal}, if a set of interconnected partial differential equations—commonly referred to as the Hamilton-Jacobi-Bellman (HJB) equations—has solutions, then a Nash equilibrium can be achieved. At this equilibrium, no player can improve their outcome by unilaterally changing their strategy. However, traditional dynamic game models often assume that all players possess complete knowledge of the game. In many real-world scenarios, players face rapidly changing and uncertain environments, leading to incomplete information about the system’s dynamics and payoffs \cite{Zhang2016,cardaliaguet2014pure,tang2021pursuit,bloembergen2015evolutionary}.

To address this uncertainty, we apply Bayesian updating methods, where players update their beliefs about unknown parameters as new information becomes available. Using the Ergodic theorem \cite{edition2002probability}, we show that as players receive independent and identically distributed signals from a stable random process, their estimates of unknown parameters converge to the true values over time, ensuring stability in decision-making.

Existing literature has explored uncertainty in dynamic game models, particularly in ecological and pollution control games, where strategies adapt to unknown ecological parameters through learning processes \cite{masoudi2016dynamic}. Studies such as \cite{mirman2012learning} and \cite{wang2023promotion} address resource management and behavior under uncertainty, typically using random differential game frameworks. However, these models often assume static beliefs about uncertain parameters, limiting their adaptability in dynamic environments \cite{zamir2020bayesian}.

In this paper, we focus on differential games with motion-payoff uncertainty, where both the system dynamics and payoff structures contain unknown parameters. Players are assumed to update their beliefs continuously using continuous Bayesian updating. This approach contrasts with traditional Bayesian updating, which typically occurs in discrete steps \cite{masoudi2016dynamic}. In continuous Bayesian updating, players revise their beliefs in real time as new signals are received \cite{zhou2024enhancing,zhou2024dynamic,zhou2024fishwars}, allowing for more accurate estimations in dynamic environments.

Moreover, while there has been substantial research on uncertainty in system dynamics, less attention has been given to uncertainties in players’ payoff functions. In classical Bayesian game frameworks \cite{zamir2020bayesian}, players’ beliefs about uncertain parameters are typically static. However, in real-world applications, as players receive new information during the game, their beliefs should evolve accordingly. This evolution of beliefs can significantly influence decision-making and strategy formulation.

To address this, we employ a continuous Bayesian updating approach, where players continually revise their uncertainties about unknown parameters based on incoming signals \cite{huang2018analysis}. We introduce a method to define players’ expected payoffs, grounded in these continuously updated beliefs. This approach addresses the limitation in classical Bayesian games where static beliefs may lead to suboptimal decisions if initial assumptions are incorrect.  A key tool in our analysis is the Kalman filter \cite{welch1995introduction,welch2021kalman}, a well-known dynamic Bayesian updating technique. We provide a proof of the convergence of players’ estimates using the Kalman filter, demonstrating its effectiveness in refining strategies in differential games with payoff uncertainty.

In Section \ref{s2}, we present the general model of differential games with motion-payoff uncertainty, including the definition of these games. We also introduce the Nash equilibrium with continuous Bayesian updating for such games. In Section \ref{s3}, we provide the exact form of the belief updating rule using the continuous Bayesian updating approach, followed by a proof that our estimators strongly converge to the true values of the unknown parameters. In Section \ref{s4}, we present simulation results for pollution control games, demonstrating the comparison between the dynamic Bayesian updating method and continuous Bayesian updating. We also derive the Nash equilibrium with continuous Bayesian updating for players in pollution control games with motion-payoff uncertainty and establish the conditions to ensure the optimal control is non-negative. Finally, in Section \ref{s5}, we offer conclusions and discuss potential directions for future research.

\section{Differential Game Model with Motion-Payoff Uncertainty}\label{s2}

This model captures two primary types of uncertainty: motion uncertainty, related to how the system's state evolves over time, and payoff uncertainty, stemming from incomplete information about the types, preferences, and strategies of other players. In this paper, we explore how these uncertainties impact player strategies and outcomes in a dynamic environment involving \(n \geq 2\) players, and analyze the robustness of strategies under these uncertainties over an infinite time horizon.

\subsection{Overview of Motion-Payoff Uncertainty in Differential Games}

The physical model in differential games captures the dynamics of the state with complete information. The evolution of the state \( S(t) \) over time is governed by the following equation:
\begin{equation}\label{continuous_state}
    \dot{S}(t) = f\left( x(t), u_1(t, \cdot), u_2(t, \cdot), \ldots, u_n(t,\cdot), S(t) \right), \quad S(t_0) = S_0,
\end{equation}
where \( S(t) \in \mathbb{R}^n \) represents the state variable at time \( t \), \( x(t)\in   \mathbb{R}\) denotes the realization of the ecological uncertainty \( X(t) \) at time \( t \), and \( u_i(t, \cdot) \) is the emission control action of the \( i \)-th player. 



In traditional game theory models, each player's payoff function is often defined under the assumption of complete information regarding the types and strategies of all other players. The payoff function for player \( i \), assuming known types \( \tau = (\tau_1, \ldots, \tau_n) \), is given by:
\begin{equation*}\label{payoff}
    J_i^\tau(u_1, \ldots, u_n; S_0) = \int_{t_0}^{\infty} g^i\left(t, u_1(t, S; 
    \tau_1,\cdot), \ldots, u_n(t, S; \tau_n,\cdot), S(t)\right) \, dt,
\end{equation*}
where \( g^i(\cdot) \) represents the instantaneous payoff for player \( i \), \( u_j(t, S;\tau_j,\cdot) \) is the control action of player \( j \) at time \( t \), and \( S(t) \) is the state of the system, evolving according to the dynamics in Equation (\ref{continuous_state}).

The existence, uniqueness,  and continuability of solution $S(t)$ for any admissible measurable controls $u_1(\cdot),\ldots, u_n(\cdot)$ was dealt with by Tolwinski, Haurie, and Leitmann~\cite{tolwinski1986cooperative}:
\begin{enumerate}
    \item $f(\cdot):\mathbb{R}\times U \times \mathbb{R}^n\rightarrow{\mathbb{R}^n}$ is continuous.

\item There exists a positive constant $k$ such that
\begin{equation*}
    \forall t \in [t_0, T] \quad and \quad \forall u \in U
\end{equation*}
\begin{equation*}
   \lVert f(x,u,S) \lVert\leq k(1+\lVert S \lVert) 
\end{equation*}

\item $\forall R>0$, $\exists K_R>0$ such that
\begin{equation*}
\forall t \in [t_0, T]\quad and  \quad \forall u \in U
\end{equation*}
\begin{equation*}
   \lVert f(x,u,S)-f(t,x',u)\lVert\leq K_R \lVert S-S' \lVert
\end{equation*}
for all $x$ and $x'$ such that
\begin{equation*}
    \lVert S \lVert \leq R\quad and  \quad   \lVert S' \lVert \leq R
\end{equation*}
\item for any $t \in [t_0,T]$ and $S \in \mathbb{R}^n$ set
   \begin{equation*}
        G(S) = \{ f(x,u,S)\vert u \in U \}
    \end{equation*}
    is a convex compact from $\mathbb{R}^n$.
\end{enumerate}

However, in practical scenarios, players typically do not have complete information about the types \( \tau_j \) of their opponents. Each player's type \( \tau_j \) may include private information such as their preferences, resource constraints, and strategic intentions. The lack of information about these types introduces epistemic uncertainty into the payoff function.

Therefore, each player must estimate the types \( \tau_j \) of their opponents based on the signals and observations available throughout the game. This estimation process affects the player's decision-making as they attempt to optimize their strategies under uncertainty. The true payoff function is thus affected by the player's belief about the unknown types, and this introduces a continuous learning component to the game. 

\subsection{Parametric Uncertainty in Motion Equations in Differential Games}

The procedure is described as follows:

At time \( t \), players observe the actual state \( S(t) \), but the realization of the random variable \( \widetilde{\eta} \) remains unobservable. Instead, players rely on their belief regarding \( \widetilde{\eta} \). Based on the observed state \( S(t) \) and the belief, player \( i \) selects a sequence of decisions for each time $k\in[t, \infty)$ to solve the subgame \( \Gamma(S(t), t, \infty) \). 

Within the defined subgame, it is assumed that the expected motion equation evolves under the belief at time \( t \), reflecting the uncertainty as perceived by the players in their estimation of the random variable \( \widetilde{\eta} \).

The following equation represents the expected system dynamics in the continuous-time subgame \(\Gamma(S(t), t, \infty)\), starting at time \(t\):
\begin{equation}\label{general_motion_equation_continuous}
    \dot{\overline{S}}(k) = f\left( \mathbb{E}[\widetilde{\eta} \vert \overline{\theta}(t)], u_1(k, \overline{S}(k); \tau_1), \ldots, u_n(k, \overline{S}(k); \tau_n), \overline{S}(k) \right), \quad \overline{S}(t) = S(t),
\end{equation}
where $k\in [t,\infty)$, \( \widetilde{\eta} \) represents the random components in the motion equations, and \( \overline{\theta}(t) \) denotes the players' estimates of the unknown parameter \( \theta \) at time \( t \).

\subsection{Parametric Uncertainty in Payoff Functions in Differential Games}

The expected payoff with continuous belief updating for player \(i\in N\), who is privy only to their own type \(\tau_i\), integrates over the probabilistic beliefs concerning the types of other players is formally captured by:
\begin{equation}\label{expected_payoff_continuous1}
\begin{split}
      &\mathbb{E}[J_{i,B}^{\tau_i}(u_1, \ldots, u_n; t, S)] =\\
      &\int_{t}^{\infty} e^{-\rho (k-t)}g^i\left[u_i(k, \overline S(k); \tau_i), \left\{ u_j(k,\overline S(k); \overline \tau_j(t))\right\}_{j \neq i}, \overline S(k)\right] \, dk,
\end{split}
\end{equation}
where $\overline S(t)= S(t)$, and \( \overline \tau_j(t)\) represents the belief parameter of player \(i\) at time \(t\) regarding the type \(\tau_j\) of player \(j\).


For the uncertainty in the payoff functions, we assume it is represented by a constant. The signals received by the players are subject to noise, and based on the historical signals up to time \( t \), players estimate the unknown parameters. The decisions at time \( t \) are based on the subgame \(\Gamma(S(t), t, \infty)\), where the state dynamics and payoff functions are defined in (\ref{general_motion_equation_continuous}) and (\ref{expected_payoff_continuous1}). 

To solve this subgame, we assume that the players' estimates of the unknown parameters for all \( k \in [t, \infty) \) remain the same as their estimates at \( k = t \), because the players only have access to historical signals up to time \( t \). After solving the subgame, setting \( k = t \) yields the players' optimal control at time \( t \). 

It is important to note that the true state trajectory is determined by observing the actual realization of the random variable at the current time.


The following definition formalizes the differential game with motion-payoff uncertainty, using continuously updated beliefs.

In the game described by equations (\ref{general_motion_equation_continuous})--(\ref{expected_payoff_continuous1}), an admissible feedback control with continuous Bayesian updating for player \( i \) takes the following form:
\[
u_i(t, S; \tau_i) = u_i(t, S; \tau_i, \{\overline{\tau}_j(t)\}_{j \in N}, \overline{x}(t)),
\]
where \( u_i\) is the control strategy for player \( i \), which depends on the state \( S \), their own type \( \tau_i \), the public belief \( \{\overline{\tau}_j(t)\}_{j \in N} \) about the types of other players, and the common belief \( \overline{x}(t) \) about the uncertainty in the motion equation.



\begin{definition}\label{def_11}
   In an $n$-player differential game with motion-payoff uncertainty, an $n$-tuple of admissible strategies
\[
\left(u_1^*(k, \overline{S}; \tau_1, \{\overline{\tau}_j(t)\}_{j \in N}, \overline{x}(t)), \dots, u_n^*(k, \overline{S}; \tau_n, \{\overline{\tau}_j(t)\}_{j \in N}, \overline{x}(t))\right)
\]
is said to constitute a generalized feedback Nash equilibrium with continuous Bayesian updating if it satisfies the feedback Nash equilibrium conditions for the subgame $\Gamma(S, t, \infty)$, as defined by (\ref{general_motion_equation_continuous}) and (\ref{expected_payoff_continuous1}).
\end{definition}

In Definition \ref{def_11}, the Nash equilibrium with continuous Bayesian updating is defined for all \( \tau_i \in T_i \) and \( u_i' \in U_i \). The equilibrium strategies depend on beliefs \( \{\overline{\tau}_j(t)\}_{j \in N} \) and the conditional expectation \( \overline{x}(t) \), computed through Bayesian updating. These beliefs and conditional expectations evolve continuously over time as players receive new information.

\begin{definition}\label{assumption_optimal_control_continuous}
The Nash equilibrium strategy with continuous Bayesian updating for player \( i \) at time \( t \) is defined as the optimal control of the subgame when \( k = t \). That is,
\[
u_i^*(t, S; \tau_i, \{\overline{\tau}_j(t)\}_{j \in N}, \overline{x}(t)) = u_i^*(k, \overline S; \tau_i, \{\overline{\tau}_j(t)\}_{j \in N}, \overline{x}(t)) \big\vert_{k=t}.
\]
\end{definition}

This Nash equilibrium extends the classical concept to continuous time, where players' strategies depend on continuous updating beliefs about unknown parameters. Each player's strategy \( u_i^* \) maximizes their expected payoff, incorporating the evolving beliefs and the estimation of other players' type, ensuring optimality given the game's current state and beliefs.


\section{Differential Game Model with Motion-Payoff Uncertainty Using Continuous Bayesian Updating}\label{s3}
After quantifying uncertainties, we explore how players utilize continuous Bayesian updating to refine their beliefs about ecological parameters based on received signals. Starting with prior beliefs, players update them to posterior beliefs as new data is acquired. Additionally, we examine the impact of informational asymmetry on payoff functions, where player \(i\) continuously updates their understanding of other players' cost parameters \(\tau_j\) based on the signals \(y_j(t)\) they receive.

\subsection{Continuous Bayesian Updating}




Formally, consider the prior belief \(\xi_t(\theta)\) at time \(t\) and a received signal \(x(t)\). The posterior belief \({\xi_t}(\theta \vert x(t))\) is updated according to Bayesian inference, and is given by:
\begin{equation}\label{Bayesian_theorem_con}
   {\xi_t}(\theta \vert x(t)) 
\propto {\phi(x(t) \vert \theta) \xi_t(\theta)},
\end{equation}
where \(\theta \in \Theta\). This formula (\ref{Bayesian_theorem_con}) delineates the updating mechanism of beliefs using the data obtained from the signal \(x(t)\). 

In stages \( t = 0, \Delta t, 2\Delta t, \ldots, \infty\), players update their joint prior distribution \( \xi_t(\theta) \) for the parameters \( \mu \) and \( \lambda \) to reflect the modified time interval \( \Delta t \). This update is based on the following conjugate prior distributions \cite{murphy2007conjugate}:
\begin{equation*}
    \xi_t(\theta) \stackrel{\text{def}}{=} \mathcal{N}\left(\mu \mid \mu({t}), \left(\kappa({t}) \lambda\right)^{-1}\right) \times \text{Ga}\left(\lambda \mid \alpha({t}), \text{rate}=\beta({t})\right),
\end{equation*}
where \( \mu(t) \), \( \kappa(t) \), \( \alpha(t) \), and \( \beta(t) \) represent the updated prior estimates or belief parameters about the mean, the precision of the normal distribution for \( \mu \), the shape, and the rate parameters of the gamma distribution for \( \lambda \), respectively. 

From the paper \cite{zhou2024enhancing}, we can directly derive the iteration rules for belief updating in the discrete-time case. In the following proposition, we transition to continuous time by considering \(\Delta t \to 0\):

\begin{proposition}\label{continuous belief motion}
In contexts of uncertainty in motion equation, players employ continuous Bayesian updating to revise their beliefs based on new data \(x(t)\) as follows:
\begin{equation}\label{belief_equation}
\begin{split}
    \dot{\mu}(t) &= \frac{x(t)-\mu(t)}{\kappa(t)+1}, \quad \mu(0) = \mu_0,\\ 
     \dot{{\kappa}}(t) &= 1, \quad {\kappa}(0) = \kappa_0,\\ 
     \dot{{\alpha}}(t) &= \frac{1}{2}, \quad {\alpha}(0) = \alpha_0, \\ 
     \dot{{\beta}}(t) &= \frac{{\kappa}(t)(x(t) - {\mu}(t))^2}{2({\kappa}(t) + 1)}, \quad {\beta}(0) = \beta_0.
\end{split}
\end{equation}    
\end{proposition}

\begin{proof}
Firstly, consider the derivative of \({\kappa}(t)\), $t\geq 0$:
\begin{equation*}
    \frac{{\kappa}(t+\Delta t) - {\kappa}(t)}{\Delta t} = 1,
\end{equation*}
and let \(\Delta t \to 0\), we obtain
\begin{equation*}
    \dot{{\kappa}}(t) = 1,
\end{equation*}
together with the initial condition, we derive
\begin{equation*}\label{k_t}
    {\kappa}(t) = t + \kappa_0.
\end{equation*}

Applying the same procedure to \({\alpha}(t)\), $t\geq 0$:
\begin{equation*}
    \lim_{\Delta t \to 0}\frac{{\alpha}(t+\Delta t) - {\alpha}(t)}{\Delta t} = \frac{1}{2},
\end{equation*}
by the definition of derivative, we find
\begin{equation}\label{alphadot}
    \dot{{\alpha}}(t) = \frac{1}{2}, \quad \alpha(0) = \alpha_0.
\end{equation}

The solution for (\ref{alphadot}) is
\begin{equation*}
    {\alpha}(t) = \frac{1}{2}t + \alpha_0.
\end{equation*}

For \({\mu}(t)\):
\begin{equation*}
   \lim_{\Delta t \to 0} \frac{{\mu}(t+ \Delta t) - {\mu}(t)}{\Delta t} = \lim_{\Delta t \to 0} \frac{x(t)-\mu(t)}{\kappa(t)+1},
\end{equation*}
we deduce that
\begin{equation*}
    \dot{\mu}(t) = \frac{x(t)-\mu(t)}{\kappa(t)+1}, \quad \mu(0) = \mu_0.
\end{equation*}

Lastly, for the differential equation of \({\beta}(t)\):
\begin{equation*}
    \lim_{\Delta t \to 0}\frac{{\beta}(t+ \Delta t) - {\beta}(t)}{\Delta t} = \lim_{\Delta t \to 0}\frac{{\kappa}(t)(x(t ) - {\mu}(t))^2}{2({\kappa}(t) + 1)},
\end{equation*}
thus, we conclude
\begin{equation*}
    \dot{{\beta}}(t) = \frac{{\kappa}(t)(x(t) - {\mu}(t))^2}{2({\kappa}(t) + 1)}.
\end{equation*}
\end{proof}

In the strategic game setting, the type of player \(i\), denoted by \(\tau_i\), is considered a constant and is common knowledge among all other players $j\in N\backslash i$. However, the exact value of \(\tau_i\) is unknown to the others. Consequently, we model the dynamics and observation of \(\tau_i\) as follows:
\begin{equation}\label{kalman}
\begin{split}
     \dot{\tau}_i(t) &= 0, \\
     y_i(t) &= \tau_i + v_i(t),
\end{split}
\end{equation}
where \(y_i(t)\) represents the signal observed by the remaining players \(j \neq i\) about the unknown parameter \(\tau_i\). The noise term \(V_i(t)\) is assumed to follow a normal distribution, specifically \(V_i(t) \sim N(0, R_i)\), where \(R_i>0\) is known to all players. 

\begin{proposition}\label{proposition_kalman_estimation}
The estimation of the unknown parameter \(\tau_i\) for player \(i\) in a continuous time setting, based on the observed signals \(y_i(t)\), can be obtained using the classical Kalman filter approach. The estimation process is governed by the following system of differential equations:
\begin{equation}\label{esti_payoff}
\begin{split}
    \dot{\overline{\tau}_i}(t) &= \frac{P_i(t)}{R_i} \left( y_i(t) - \overline{\tau}_i(t) \right), \quad \overline{\tau}_i(0) = \tau_i^0,\\
    \dot{P}_i(t) &= -\frac{(P_i(t))^2}{R_i}, \quad P_i(0) = P_i^0,
\end{split}
\end{equation}
where \(\overline{\tau}_i(t)\) represents the estimated value of the unknown parameter \(\tau_i\) at time \(t\), and \(P_i(t)\) is the estimation error covariance at time \(t\). 
\end{proposition}
\begin{proof}
    The result follows directly from the classical continuous Kalman filter formulation and can be derived in a straightforward manner using standard continuous Kalman filter equations.
\end{proof}

The Kalman gain \(\frac{P_i(t)}{R_i}\) dynamically modulates this adjustment, ensuring that the estimation becomes more accurate as new information is received.

\subsection{Strong Belief Convergence Theorems under Motion Equation
Uncertainty}


In this chapter, \(M(t)\) represents the estimator for the mean \(\mu(t)\) of the random variable \(\widetilde{\eta}\). As \(t \to \infty\), \(M(t)\) converges to the true mean \(\mu\), indicating that as more data is collected over time, the estimate becomes increasingly accurate.

\begin{proposition}\label{ergodic}
   Let \( \{X(t): t \in \mathcal{I}\} \) be a random process where each \( X(t) \) is an independent and identically distributed (i.i.d.) random variable, and each \( X(t) \) follows a normal distribution \( N(\mu, \sigma^2) \). Then the process \( \{X(t)\} \) is both strictly stationary and ergodic, where \( \mathcal{I} \) represents the index set over which the process is defined, such as continuous time \( \mathcal{I} = \mathbb{R} \) or discrete time \( \mathcal{I} = \mathbb{Z} \).
\end{proposition}
\begin{proof}
  Since each \( X(t) \) is i.i.d. and follows the same distribution \( N(\mu, \sigma^2) \), the distribution of any finite collection \( \{X(t_1), X(t_2), \dots, X(t_k)\} \) is always the same, regardless of the specific times \( t_1, t_2, \dots, t_k \). Thus, \( \{X(t)\} \) is strictly stationary \cite{park2018fundamentals} because:
   \[
   (X(t_1), X(t_2), \dots, X(t_k)) \overset{d}{=} (X(t_1+h), X(t_2+h), \dots, X(t_k+h)) \quad \text{for all } h,
   \]
   where \( \overset{d}{=} \) denotes equality in distribution.

   Ergodicity in this context means that the time average of a single realization converges almost surely to the expected value of the process \cite{edition2002probability}. For each \( t \), since \( X(t) \sim N(\mu, \sigma^2) \) and each \( X(t) \) is independent of others, by the Strong Law of Large Numbers (SLLN) for i.i.d. random variables:
   \[
   \lim_{V \to \infty} \frac{1}{V} \int_0^V X(t) = \mu \quad \text{almost surely}.
   \]

Therefore, the random process \( \{X(t)\} \) is both strictly stationary and ergodic.
\end{proof}

 \begin{theorem}\label{thm estimation of unknown mean1}
 (Strong convergence)
     The estimator of the unknown mean will tend to the real value, which means 
     \begin{equation*}
        \lim _{t \rightarrow \infty} M(t)=\mu, 
     \end{equation*}
    where $M(t)$ is the random variable describing the related belief $\mu(t)$ at stage $t>0$; it is the belief of the unknown mean $\mu$ at stage $t$.
 \end{theorem}
\begin{proof}
From  (\ref{belief_equation}), we start by analyzing the differential equation:
\begin{equation}\label{eq1}
\dot{M}(t) = \frac{X(t) - M(t)}{k_0 + t+1}.
\end{equation}

We first examine the homogeneous part of \eqref{eq1}:
\begin{equation*}
\dot{M}(t) = -\frac{M(t)}{k_0 + t+1}.
\end{equation*}
Solving this, we find:
\begin{equation*}
M(t) = \frac{c}{k_0 + t+1}.
\end{equation*}

Next, let \(M(t) = \frac{v(t)}{k_0 + t+1}\) and substitute back into \eqref{eq1} to obtain:
\begin{equation*}
v(t) = \int_0^t X(\tau) \, d\tau + c.
\end{equation*}

With the initial condition \(M(0) = \mu_0\), it follows that \(c = \mu_0 k_0\). Therefore:
\begin{equation*}
M(t) = \frac{1}{k_0 + t+1} \left(\int_0^t X(\tau) \, d\tau + \mu_0 k_0\right).
\end{equation*}

As \(t \to \infty\), we consider:
\begin{equation*}
    \lim_{t \to \infty} M(t) = \lim_{t \to \infty} \frac{t}{t + k_0+1} \left(\frac{1}{t}\int_0^t X(\tau) \, d\tau\right) + \frac{\mu_0 k_0}{t + k_0+1} = \mu.
\end{equation*}


This proof demonstrates that the estimator \(M(t)\) converges to the true unknown mean of the distribution over time.
\end{proof}

\begin{proposition}\label{finite}
    Let \(X(t) \sim N(\mu(t), \sigma^2(t))\) be a sequence of independent normally distributed random variables for \(t \geq 0\), where \(\int_{0}^\infty \mu(t) \, dt < \infty\) and \(0\neq \int_{0}^\infty \sigma^2(t) \, dt < \infty\). Then the random integral \(\int_{0}^\infty X(t) \, dt\) converges almost surely to \(\int_{0}^\infty \mu(t) \, dt\).
\end{proposition}

\begin{proof}

We begin by defining a new random variable \(Z = \int_{0}^\infty X(t) \, dt\) and analyze its expectation and variance.

\begin{equation}\label{expec}
    \mathbb{E}[Z] = \mathbb{E}\left[\int_{0}^\infty X(t) \, dt\right] = \int_{0}^\infty \mathbb{E}[X(t)] \, dt = \int_{0}^\infty \mu(t) \, dt<\infty.
\end{equation}

\begin{equation}\label{var}
\mathbb{D}[Z] = \mathbb{D}\left[\int_{0}^\infty X(t) \, dt\right] = \int_{0}^\infty \mathbb{D}[X(t)] \, dt = \int_{0}^\infty \sigma^2(t) \, dt<\infty.  
\end{equation}
The equality between the second and third expressions of (\ref{expec}) and (\ref{var}) in the equations rely on Fubini's Theorem \cite{Fubini1907}.


Applying Chebyshev's Inequality \cite{Feller1968}, for any \(\epsilon > 0\), we have:
\[
{P}\left(\left\vert Z - \mathbb{E}[Z]\right\vert \geq \epsilon\right) \leq \frac{\mathbb{D}[Z]}{\epsilon^2}.
\]
Since $\mathbb{D}[Z]<\infty$, this probability becomes arbitrarily small as \(\epsilon\) grows larger.

Define events \(A_n\) for \(n = 1, 2, 3, \ldots\) as follows:
\[
A_n = \left\{\left\vert Z - \mathbb{E}[Z]\right\vert \geq \frac{1}{n}\right\}.
\]
Using Chebyshev's Inequality, we obtain:
\[
{P}(A_n) \leq n^2 \mathbb{D}[Z].
\]
Since \(\sum_{n=1}^\infty \frac{1}{n^2}\) is a convergent series and \(\mathbb{D}[Z]<\infty\), the series \(\sum_{n=1}^\infty {P}(A_n)<\infty\).

By the Borel-Cantelli Lemma \cite{Borel1909},\cite{Cantelli1917}, this implies that
\[
{P}\left(\limsup_{n \to \infty} A_n\right) = 0.
\]
This indicates that:
\[
\left\vert\int_{0}^\infty X(t) \, dt - \int_{0}^\infty \mu(t) \, dt\right\vert < \frac{1}{n} \text{ eventually for almost all } n.
\]

As $n$ goes to infinity, $\frac{1}{n}$ approaches zero. Hence, the absolute difference between the random integral and its expected value becomes arbitrarily small:
\[\left\vert\int_{0}^{\infty} X(t) d t-\int_{0}^{\infty} \mu(t) d t\right\vert<\epsilon,\]
 for any $\epsilon>0$.

Hence,
\[
\int_{0}^\infty X(t) \, dt \text{ converges almost surely to } \int_{0}^\infty \mu(t) \, dt.
\]
\end{proof}

\begin{theorem}\label{variance in motion}
     The variance of the estimator of the unknown mean almost surely tends to 0. In other words, 
          \begin{equation}\label{variance}
              \lim _{t \rightarrow \infty} \frac{B({t})}{\kappa({t})\left(\alpha({t})-1\right)}=0,
          \end{equation}
          where $\frac{B({t})}{\kappa({t})(\alpha({t})-1)}$ is a random variable representing the variance of the estimator of the unknown mean at the time $t>0$.
\end{theorem}
\begin{proof}
  The proof is given in \ref{appendix:proof2}.

\end{proof}

\subsection{Strong Belief Convergence Theorems under Payoff Functions Uncertainty}

\begin{theorem}\label{variance in cost1}
For the Kalman filter modeled by (\ref{kalman}), the uncertainty \( P_i(t) \) in the belief about \( \tau_i(t) \) vanishes over time:
\begin{equation*}
    \lim_{t \to \infty} P_i(t) = 0, \quad \forall i \in N,
\end{equation*}
where \( P_i(t) \) quantifies the uncertainty in estimating \( \tau_i(t) \) based on available information at time \( t \).
\end{theorem}
\begin{proof}
We analyze the differential equation governing the estimation error variance:
\begin{equation}\label{eq3}
    \dot{P}_i(t) = -\frac{(P_i(t))^2}{R_i},
\end{equation}
where \( R_i \) denotes the constant variance of the measurement noise.

Transforming and integrating the differential equation (\ref{eq3}):
\begin{equation*}\label{eq2}
    \frac{dP_i(t)}{P_i(t)^2} = -\frac{1}{R_i} dt.
\end{equation*}
Integrating from 0 to \( t \) with the initial condition \( P_i(0) = P_i^0 \) leads to:
\begin{equation*}
    P_i(t) = \frac{P_i^0 R_i}{tP_i^0 + R_i}.
\end{equation*}
This yields:
\begin{equation*}
    \lim_{t \to \infty} P_i(t) = 0.
\end{equation*}
\end{proof}

\begin{theorem}\label{thm estimation of unknown cost1}
For the Gaussian process \( Y_i(t) \) where \( t \geq 0 \), the estimated value \( \overline{T}_i(t) \), representing the prior estimation of \( \tau_i \) at time \( t \), converges to the true unknown value \( \tau_i \):
\begin{equation*}
    \lim_{t \to \infty} \overline{T}_i(t) = \tau_i, \quad \forall i \in N.
\end{equation*}
\end{theorem}
\begin{proof}
Given the dynamic estimation model (\ref{esti_payoff}), it evolves according to:
\begin{equation}\label{eq8}
    \dot{\overline {T}}_i(t) = \frac{P_i(t)}{R_i}(Y_i(t) - \overline{T}_i(t)).
\end{equation}
Solving (\ref{eq8}), we obtain:
\[
\overline{T}_i(t) = \frac{P_i^0 \int_0^t Y_i(s) \, ds}{t P_i^0 + R_i}.
\]
As \( t \) approaches infinity:
\[
\lim_{t \to \infty} \overline{T}_i(t) = \lim_{t \to \infty} \frac{P_i^0 \frac{\int_0^t Y_i(s) \, ds}{t}}{P_i^0 + \frac{R_i}{t}} = \tau_i,
\]
where the ergodic property in Proposition \ref{ergodic} of \( Y_i(s) \) implies its time average converges to \(\tau_i\).
\end{proof}

\section{Pollution Control in Continuous Time with Motion-Payoff Uncertainty}\label{s4}
Having established the continuous-time differential game model, we now focus on quantifying the motion and payoff uncertainties in pollution control. 

\subsection{Uncertainty Quantification}
Pollution accumulation over time can be modeled in a continuous framework, enhancing our previous analyses in discrete-time dynamics as discussed in \cite{zhou2024enhancing}. In this extended model, we consider a pollution control game involving \(n\) strategic players.

We denote by \(\overline S(t)\) the stock of pollution at time \(t\in [0,\infty)\), which is contributed to by various countries affecting the same pollution stock. 

The evolution of \(\overline S(k)\) of subgame $\Gamma(S(t), t, \infty)$ is governed by the following linear differential equation:
\begin{equation}\label{motion_continuous}
\dot{\overline S}(k) = \overline x(t) \sum_{i=1}^n u_i(k,\overline S(k);\tau_i) - \left(1 - \overline x(t) \delta\right) \overline S(k),\quad    \overline S(t)=S(t), 
\end{equation}
where \( 1-\delta \) represents the proportion of pollution that decays over a unit of time. The term \( \overline{x}(t) = \mathbb{E}[\widetilde{\eta} \vert \overline{\theta}(t)] \) denotes the conditional expectation of the random variable \( \widetilde{\eta} \) at time \( t \). 

Firstly, we will provide an economic interpretation of our proposed motion equation, considering both the decay rate ($1-\delta$) and absorption rate ($\overline x$) simultaneously in continuous time.

The second term on the right-hand side of equation (\ref{motion_continuous}), \( \left( 1 - \overline{x}(t) \delta \right) \overline{S}(t) \), represents the rate at which pollution naturally disappears. Here, \( \delta \) is the fraction of pollution that remains in nature, while \( \overline{x}(t) \) is the pollution absorption rate. The term \( 1 - \overline{x}(t) \delta \) reflects the portion of pollution that disappears over time, with the negative sign indicating a reduction in pollution, scaled by the current pollution level \( \overline{S}(t) \).

We extend our investigation to a continuous-time framework for a pollution control game involving $n\geq 2$ players, each facing uncertainty about the other players' types except their own. 

 The expected payoff with continuous belief updating for player \(i\), who knows only his type but not those of others, is given by:
\begin{equation}\label{expected_payoff_continuous}
\begin{split}
    &  \mathbb{E}[J_{i,B}^{\tau_i}(u_1, \ldots, u_n; t, S(t))] =\\
     & \int_{t}^{\infty} e^{-\rho (k-t)}\left [u_{i}(k, \overline S(k); \tau_i)\left(a_i-u_{i}-\sum_{j\not = i}^{n} {u_{j}(k, \overline S(k); \overline \tau_j(t) )}\right)-{\tau_i}\overline S(k) \right ] dk, 
\end{split}
\end{equation}
where \(\overline \tau_j(t)=\int_{ T_j} p_i^t(\tau_j\vert {y_j(t)}) \cdot  \tau_j d \tau_j\) represents the belief parameters of player \(i\) at time \(t\) about the type \(\tau_j\) of player \(j\) based on the signal \(y_j(t)\).


\subsection{Constructing Nash Equilibrium with Continuous Bayesian Updating in Pollution Control}

\begin{proposition}\label{optimal control}
The feedback Nash equilibrium with continuous Bayesian updating 
\[
\left(u_1^*(t, S; \tau_1, \{\overline{\tau}_j(t)\}_{j \in N}, \overline{x}(t)), \dots, u_n^*(t, S; \tau_n, \{\overline{\tau}_j(t)\}_{j \in N}, \overline{x}(t))\right)
\]
for the pollution control differential game with motion-payoff uncertainty defined by (\ref{motion_continuous})--(\ref{expected_payoff_continuous}) is given by:
\begin{equation}\label{ne_cbu}
\begin{split}
    &u_i^*(t, S; \tau_i, \{\overline{\tau}_j(t)\}_{j \in N}, \overline{x}(t)) =\\ &a_i - \frac{a}{n+1} - \frac{n^2 - n + 2}{4(n+1)} \overline{c}(t) \sum_{j=1}^n \overline{\tau}_j(t) + \frac{1}{2} \overline{c}(t) \left( \frac{n}{2} \overline{\tau}_i(t) + \tau_i \right),
\end{split}
\end{equation}
where $a = \sum_{i=1}^n a_i$, $\overline c(t)=-\frac{\overline x(t)}{1- \overline x(t)\delta-\rho}$.

To ensure the non-negativity of the optimal control, \( u_i^*(t, \cdot) \geq 0 \), the following two sufficient conditions must hold:

1. For the parameter \(a_i\):
\[
\underline{a} - \frac{n}{n+1} \overline{a} > 0,
\]
where \( \underline{a} = \min\limits_{i \in \{1,2,..,n\}} a_i \) and \( \overline{a} = \max\limits_{i \in \{1,2,..,n\}} a_i \).

2. For the type \( \tau_i \):
\[
-\frac{n^2 - n + 2}{4(n+1)} n \underline q + \left( \frac{n}{2} + 1 \right) \overline{Q} < 0,
\]
where \( \underline{q} = \min\limits_{i \in \{1,2,..,n\}} q_i \) and \( \overline{Q} = \max\limits_{i \in \{1,2,..,n\}} Q_i \).

3. For the discount factor
\[
1-\rho\geq\delta>0.
\]
\end{proposition}
\begin{proof}
The proof is given in \ref{appendix:proof3}.
\end{proof}

To compare the Nash equilibrium with continuous Bayesian updating in the pollution control game under uncertainty with the Nash equilibrium in the classical pollution control game, where all parameters are known. The known-state Nash equilibrium, in the absence of uncertainties, is defined for each player as:
\begin{equation}
   u_i^{*,known}(\{\tau_j\}_{j \in N}, t) = a_i - \frac{a}{n+1} - \frac{n^2 - n + 2}{4(n+1)}  c \sum_{j=1}^n \tau_j + \frac{n+2}{4}  c \tau_i,
\end{equation}
where \( c \) is a constant defined as \(  c= -\frac{\mu}{1 - \mu \delta-\rho} \). In this case, the optimal control does not depend on the current state, as we consider a linear-state game; hence, the state is omitted from the parameters of the optimal control expression.

In the pollution control model with motion-payoff uncertainty, the optimal control for each player \( i \) at time \( t \) is independent of the current state and system state. Therefore, the Nash equilibrium with continuous Bayesian updating can be written as \( u_i^*(\tau_i, \{\overline{\tau}_j(t)\}_{j \in N}, \overline{x}(t)) \). We now prove the convergence of the optimal control to the known-state Nash equilibrium control. Specifically, we show that:
\begin{theorem}\label{overall convergence}
As the time progresses and the players' beliefs converge, the Nash equilibrium with continuous Bayesian updating approaches the known-state optimal control. Specifically, we have the following convergence result:
\begin{equation}\label{control_convergen}
    \lim_{t \to \infty} u_i^*(\tau_i, \{\overline{\tau}_j(t)\}_{j \in N}, \overline{x}(t)) = u_i^{*,known}(\{\tau_j\}_{j \in N}, t).
\end{equation}
\end{theorem}
\begin{proof}
   As we have previously established that \( \overline{\tau}_i(t) \to \tau_i \) and \( \overline{c}(t) \to c(t) \), it follows that the optimal control converges to the known-state Nash equilibrium control. Therefore, we can easily obtain equation (\ref{control_convergen}).
\end{proof}
\subsection{Comparative Simulations of Dynamic and Continuous Bayesian Updating}

We use a step function to define how the signals received by players at discrete time points are applied to continuous time. In this scenario, we consider a time interval of length 10, and the time step for new signals, denoted as $\Delta t$, is 0.02. Therefore, players will receive 500 signals regarding each type of uncertainty.

\begin{figure}[ht] 
    \centering    
    \subfloat[Evolution of ecological uncertainty signals] 
    {
        \begin{minipage}[t]{0.5\textwidth}
            \centering          
            \includegraphics[width=0.8\textwidth]{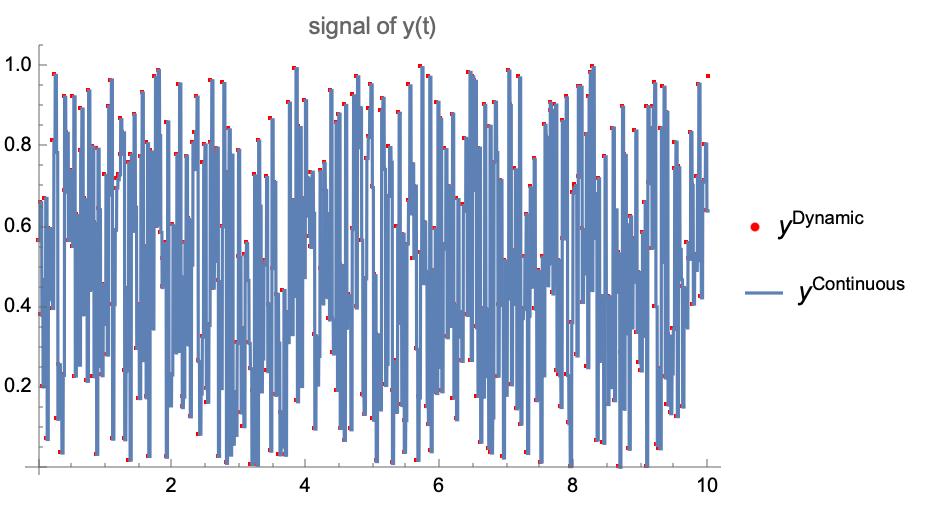} 
            \label{fig1.1}
        \end{minipage}%
    }
    \subfloat[Reception of the control cost-related signal] 
    {
        \begin{minipage}[t]{0.5\textwidth}
            \centering      
            \includegraphics[width=0.8\textwidth]{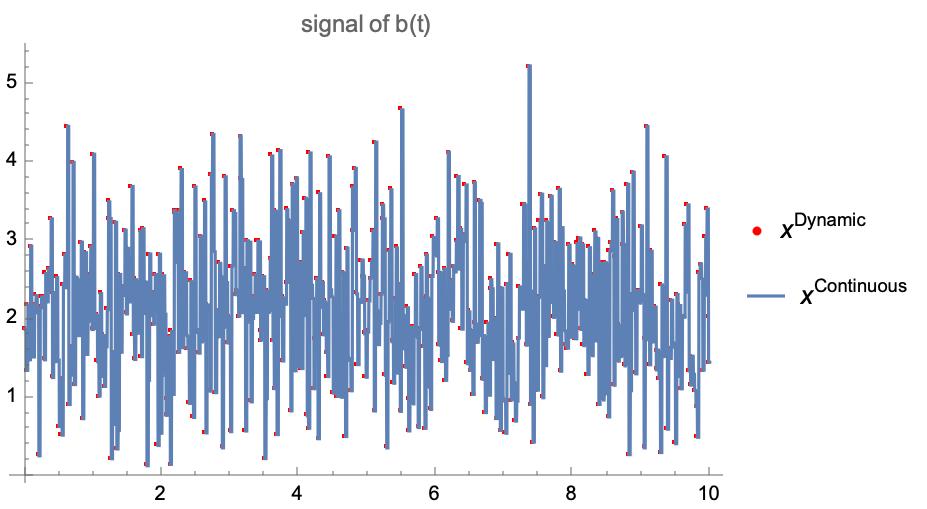}   
            \label{fig1.2}
        \end{minipage}
    }%
    \caption{Signal analysis in pollution control games} 
    \label{fig1}  
\end{figure}

Figure \ref{fig1.1} illustrates the actual observations of ecological uncertainty as received by the players at each moment, while Figure \ref{fig1.2} displays the noisy signals about the cost parameter $\tau_i$, $i\in N$, derived from sources such as market data or environmental measurements, which often embody inherent uncertainties. In both figures, red dots represent signals received by players at discrete moments, and the blue lines indicate signals received over continuous time.

We employ a step function approach to depict the relationship between signals received at discrete times and their representation over continuous time intervals. Specifically, a signal received at a discrete moment \( t \) is modeled as being continuously received over the interval \([t, t + \Delta t]\), for any \( t \in [0, \infty) \), where \( \Delta t > 0 \) is a positive duration.

\begin{figure}[H] 
    \centering    
    \subfloat[Comparison of ecological uncertainty estimation in motion equation] 
    {
        \begin{minipage}[t]{0.5\textwidth}
            \centering          
            \includegraphics[width=0.8\textwidth]{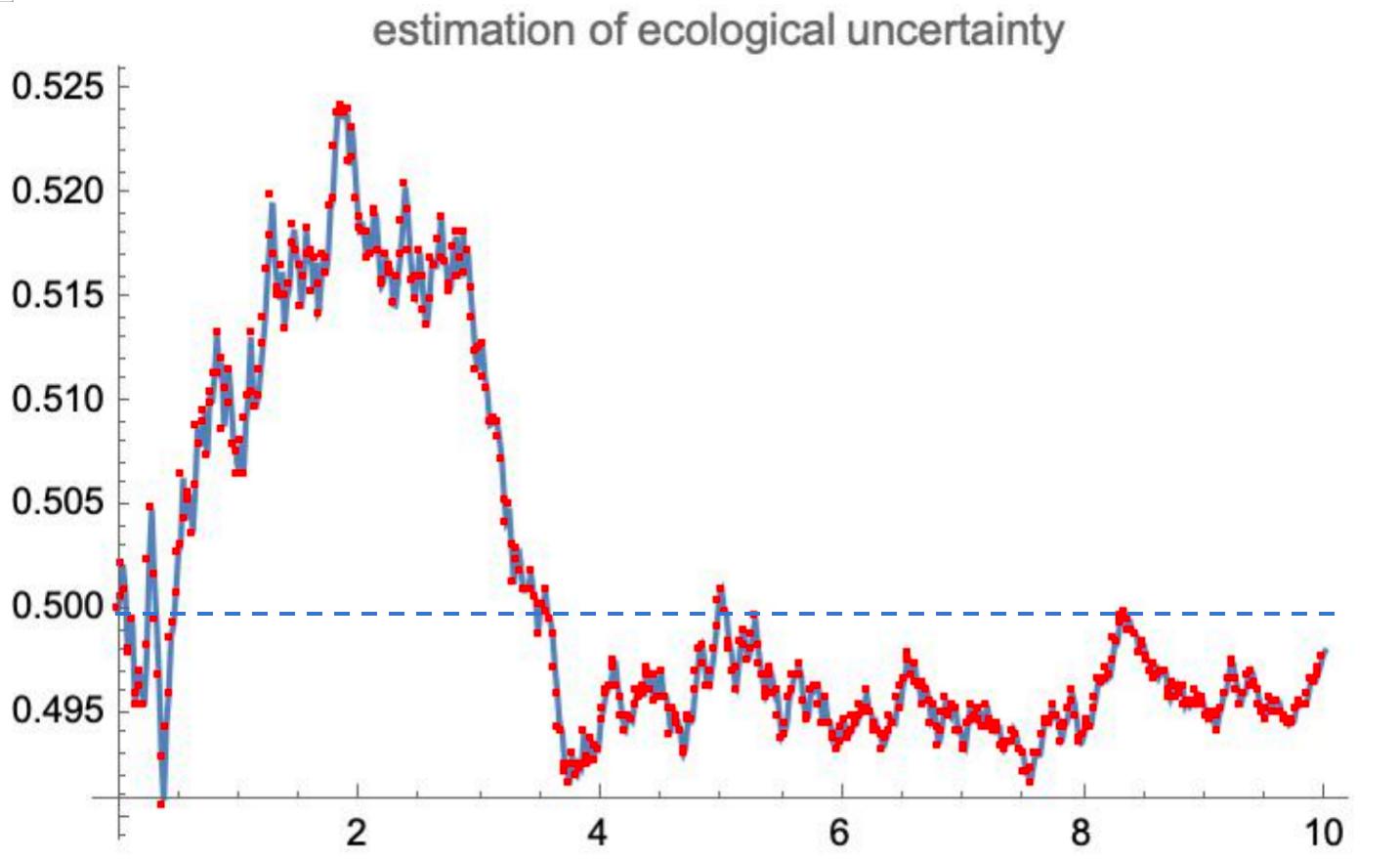} 
            \label{fig2.1}
        \end{minipage}%
    }
    \subfloat[Comparison of player's estimation of other player's cost parameter] 
    {
        \begin{minipage}[t]{0.5\textwidth}
            \centering      
            \includegraphics[width=0.8\textwidth]{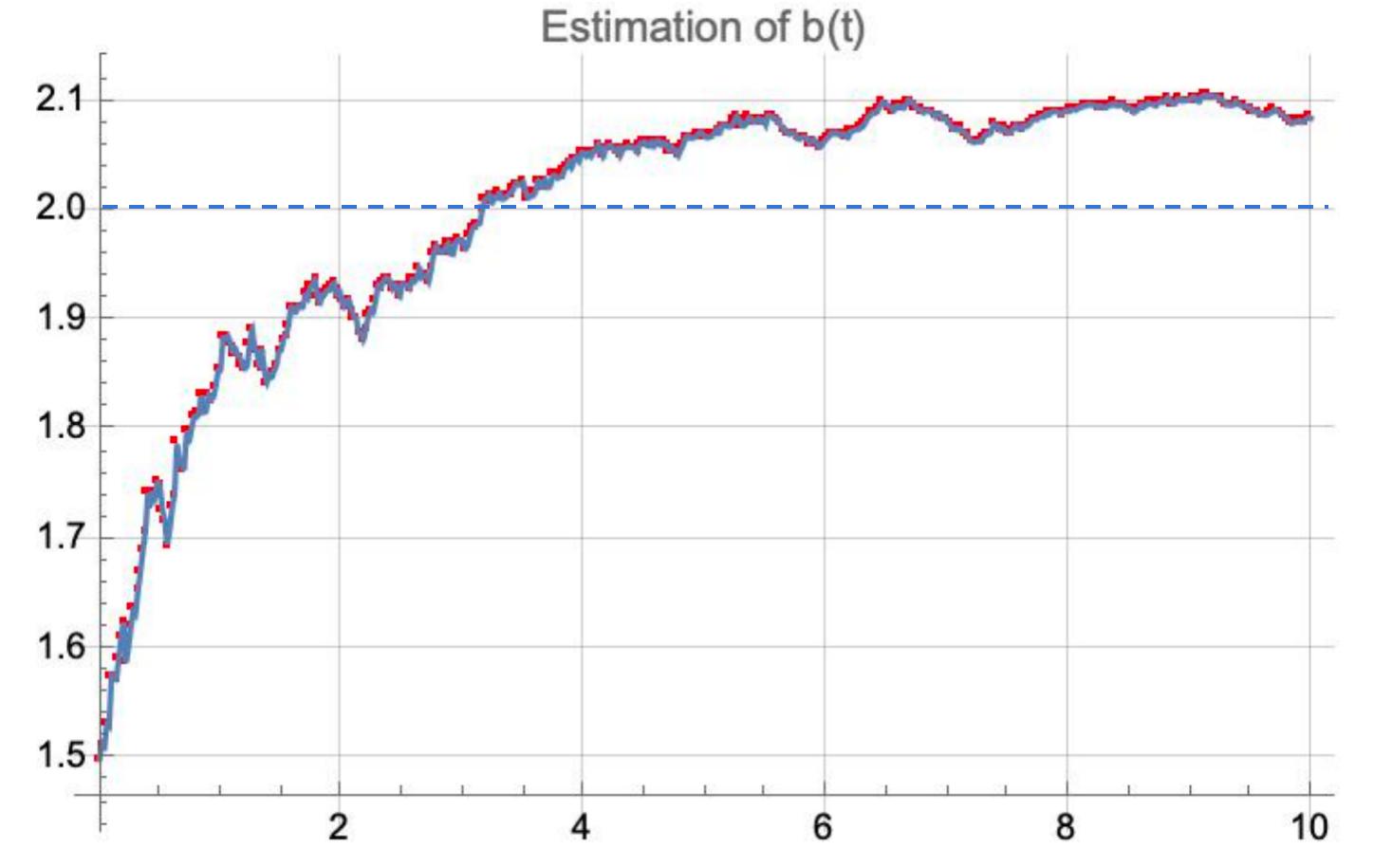}   
            \label{fig2.2}
        \end{minipage}
    }%
    
    \caption{Analysis of dynamic system parameter estimation} 
    \label{fig2}  
\end{figure}

In Figs. \ref{fig2.1} and \ref{fig2.2}, the player's assessments of ecological uncertainty and pollution management costs converge closely as the time interval \(\Delta t\) decreases. This alignment is achieved through the application of Bayesian updating in both continuous and discrete settings, particularly with a small interval of 0.02. Such convergence underscores the efficacy of Bayesian updating methods in ensuring consistent estimates across different time scales.

\begin{figure}[H] 
    \centering    
    \subfloat[Estimation of ecological Uncertainty] 
    {
        \begin{minipage}[t]{0.5\textwidth}
            \centering          
            \includegraphics[width=0.8\textwidth]{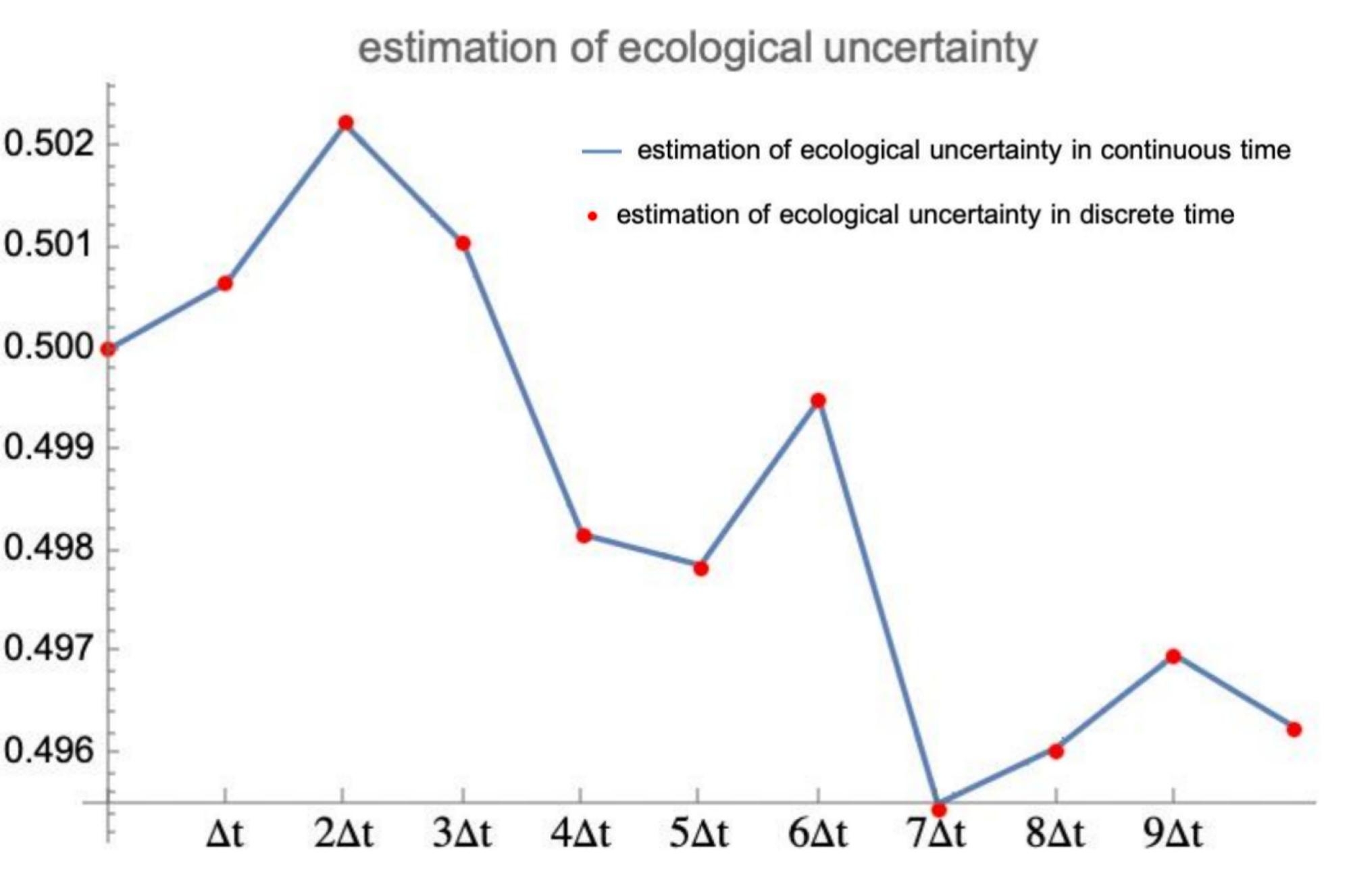} 
            \label{fig2.11}
        \end{minipage}%
    }
    \subfloat[Estimation of other players' cost parameters] 
    {
        \begin{minipage}[t]{0.5\textwidth}
            \centering      
            \includegraphics[width=0.8\textwidth]{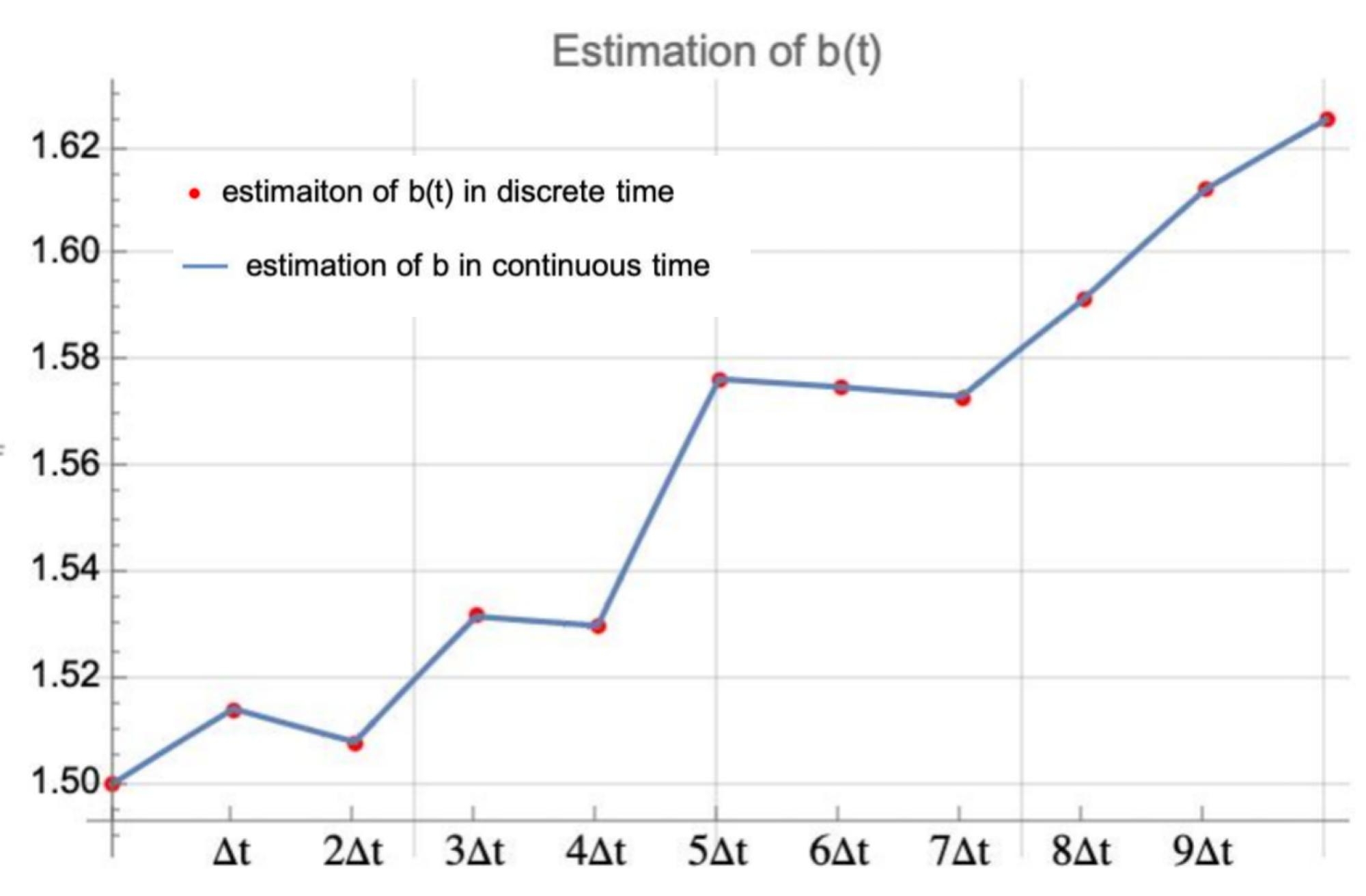}   
            \label{fig2.22}
        \end{minipage}
    }%
    
    \caption{Detailed comparative analysis of parameter estimations} 
    \label{fig22}  
\end{figure}

To more clearly illustrate the results, we present a magnified section in Figs. \ref{fig2.11} and \ref{fig2.22}, showing outcomes from 0 to $10\Delta t$. The zoomed-in view reveals that, when the interval between received signals is sufficiently small, players' beliefs about unknown parameters align closely in both continuous and discrete-time settings under continuous Bayesian updating and dynamic approaches. 

\begin{figure}[H] 
    \centering    
    \subfloat[Variance in estimation of ecological uncertainty] 
    {
        \begin{minipage}[t]{0.5\textwidth}
            \centering          
            \includegraphics[width=0.8\textwidth]{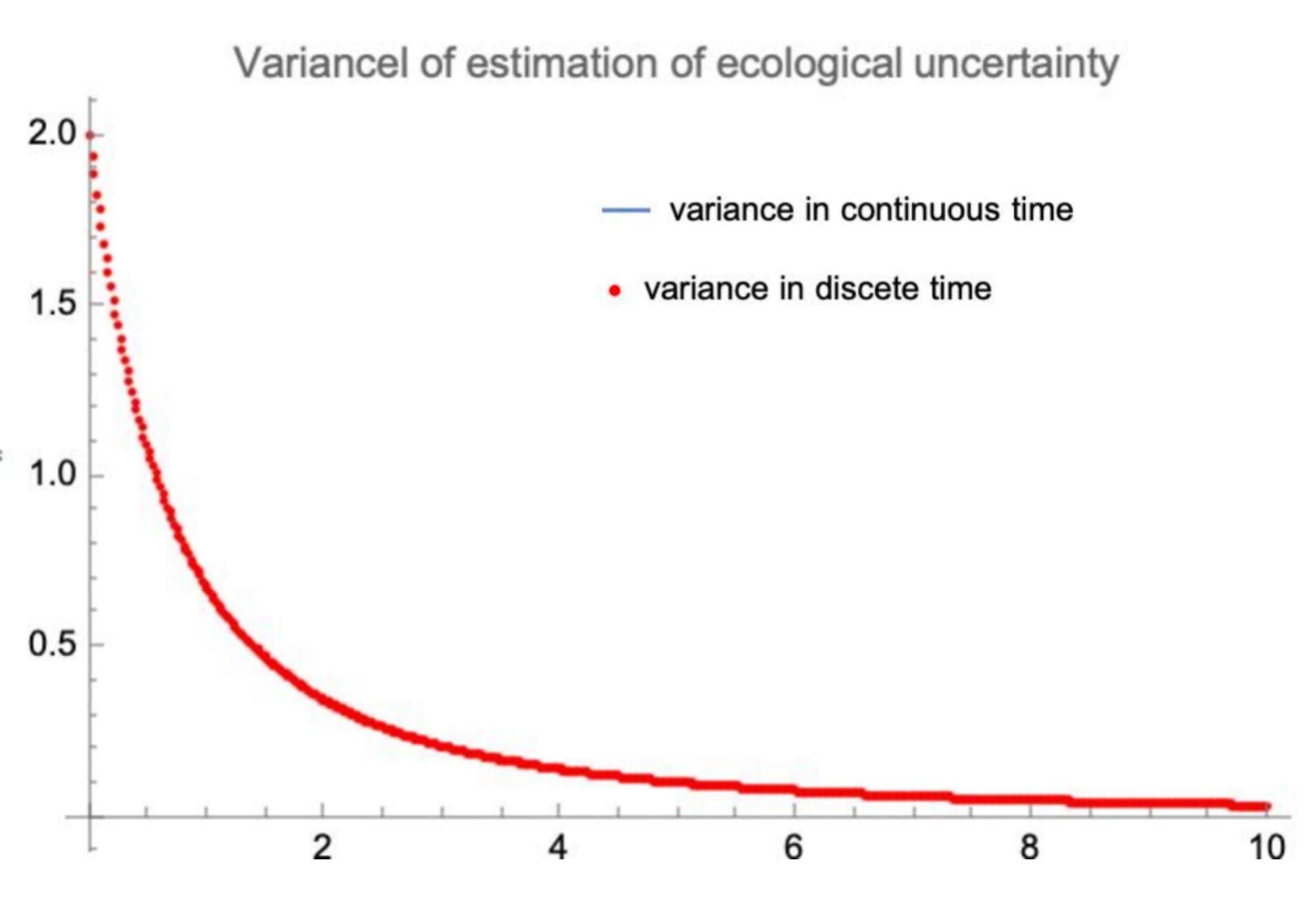} 
            \label{fig3.1}
        \end{minipage}%
    }
    \subfloat[Variance in player's cost estimation] 
    {
        \begin{minipage}[t]{0.5\textwidth}
            \centering      
            \includegraphics[width=0.8\textwidth]{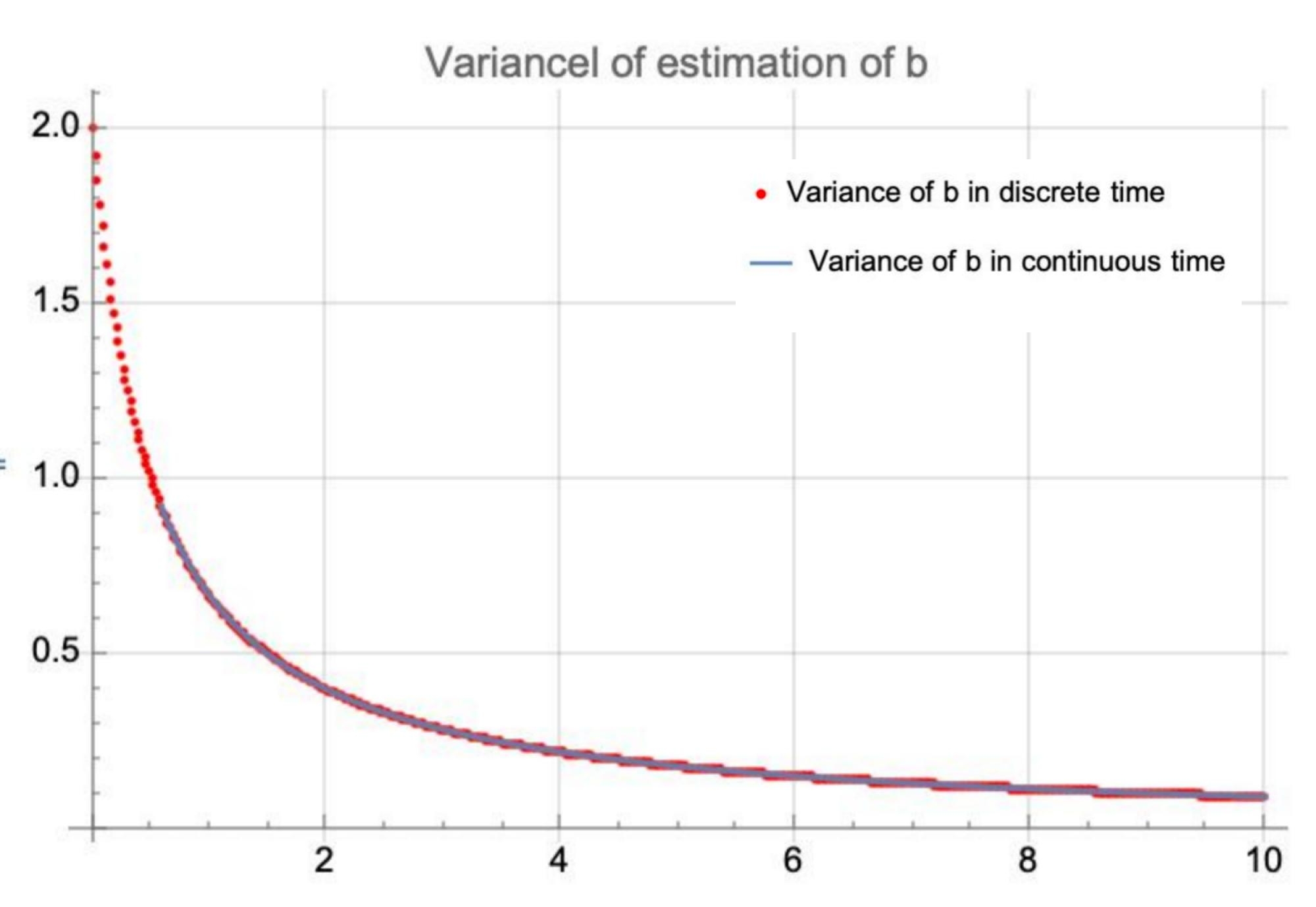}   
            \label{fig3.2}
        \end{minipage}
    }%
    
    \caption{Advanced variance analysis in ecological uncertainty and cost estimation} 
    \label{fig3}  
\end{figure}
In both Fig. \ref{fig3.1} and Fig. \ref{fig3.2}, there's a notable and consistent trend observed-over time, the range of variance in the players' estimates steadily narrows, eventually approaching zero. This trend suggests that in pollution control games, as players continually observe and adapt based on new information, their predictions concerning environmental challenges and pollution control costs become increasingly precise, moving towards a shared understanding. 

\begin{figure}[H]
\centering
\includegraphics[width=0.5\linewidth]{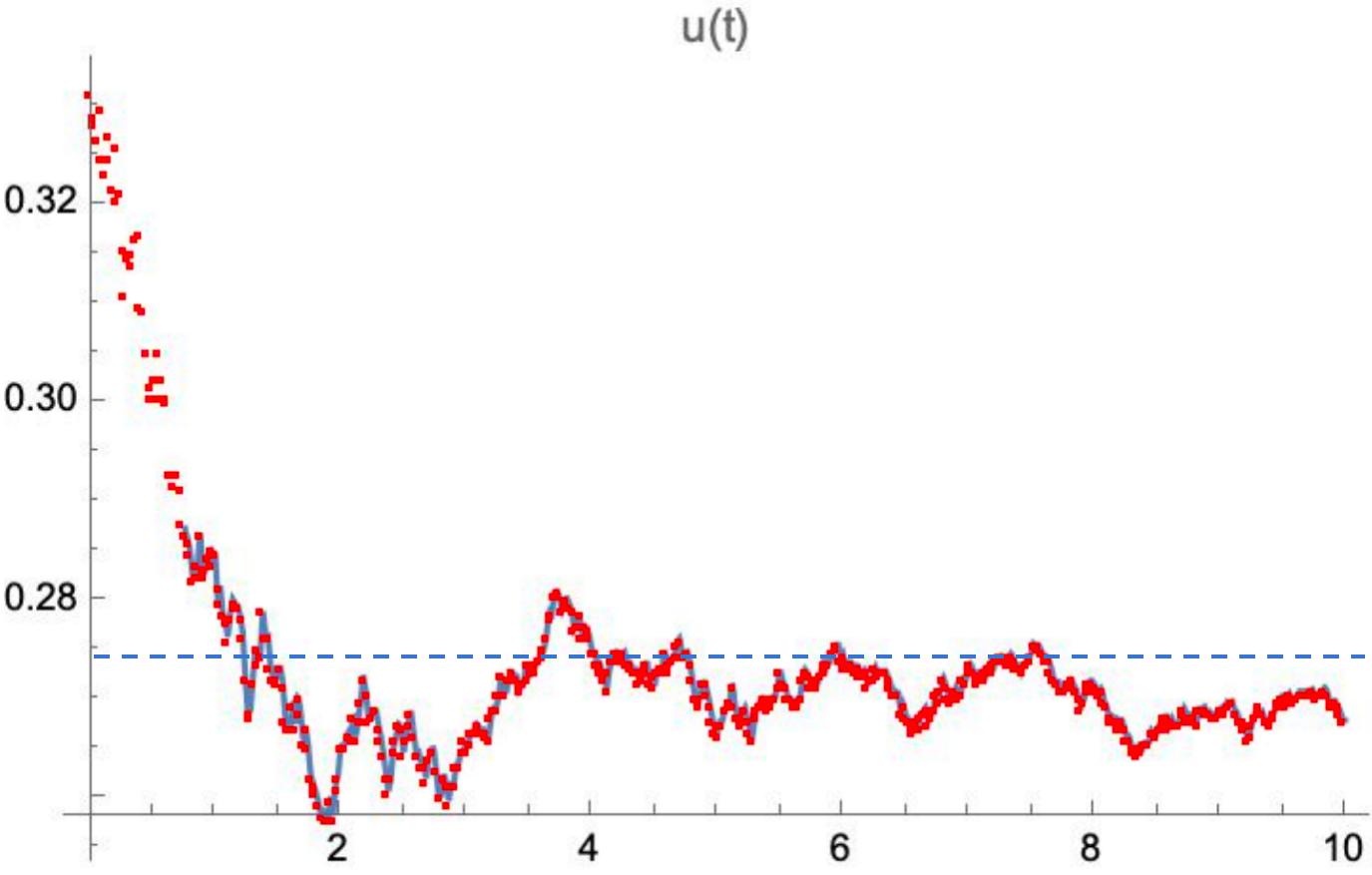}
\caption{Comparative controls for continuous vs. discrete time}
\label{fig4}
\end{figure}

From Fig. \ref{fig4}, we can observe that the pollution emissions by players, who receive signals about unknown parameters discretely, will align with those in continuous scenarios. This convergence occurs as players receive sufficient signals over time, allowing their estimates of the unknown parameters to gradually approximate the true values, regardless of whether the signals are received discretely or continuously. 

\begin{figure}[H]
\centering
\includegraphics[width=0.5\linewidth]{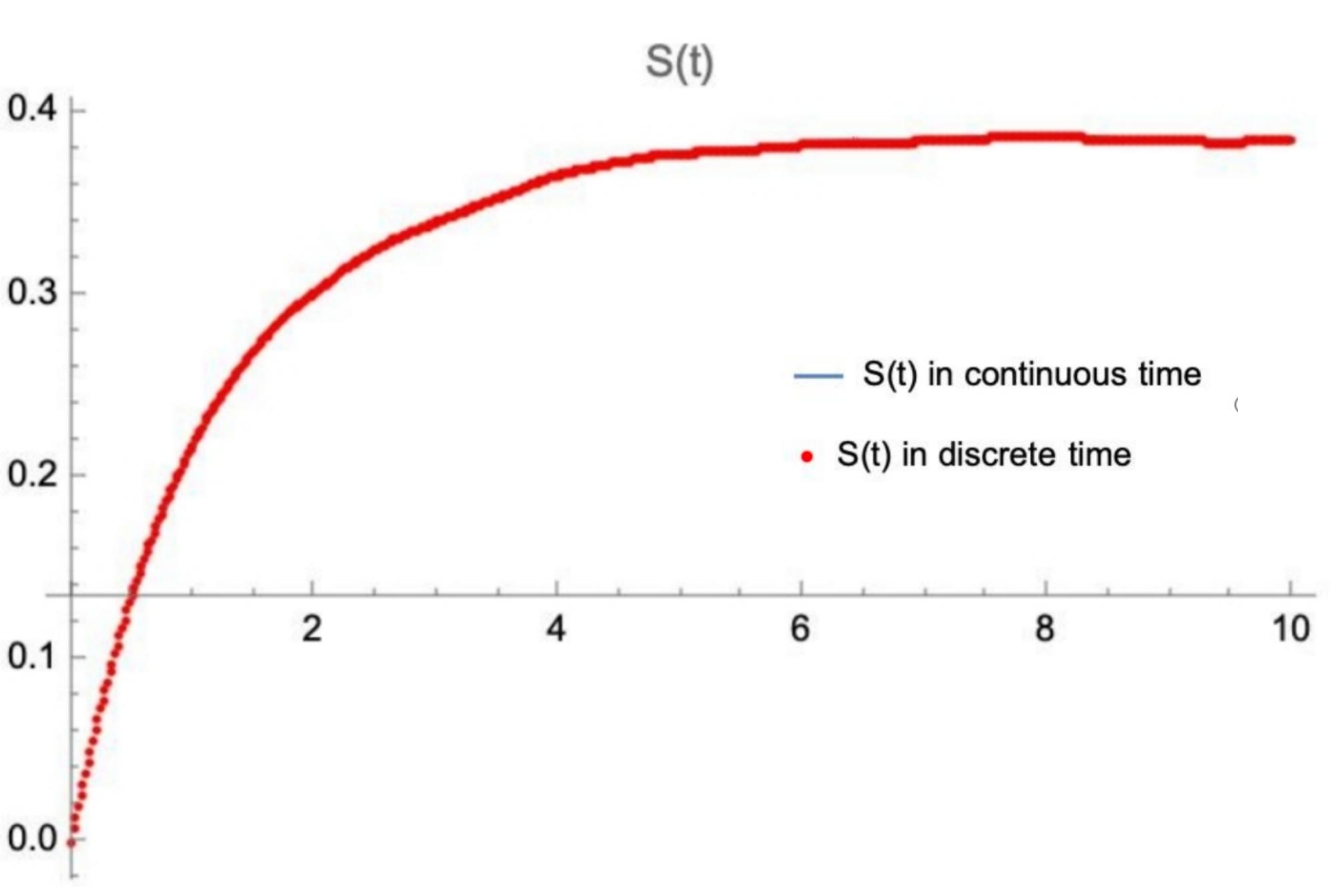}
\caption{Comparative trajectories for continuous vs. discrete time}
\label{fig5}
\end{figure}

In Fig. \ref{fig5}, we observe that the pollution storage levels in both discrete and continuous time settings tend to align closely over time. This alignment occurs because the intervals at which players receive signals in discrete settings are sufficiently small, allowing for consistent estimates of unknown parameters across both discrete and continuous frameworks. Consequently, this consistency in parameter estimation leads to similar optimal controls, ensuring that the state trajectories in both scenarios are highly congruent.

\section{Conclusions}\label{s5}
In this paper, we have explored the motion equation and payoff functions in continuous-time pollution control games under uncertainties. By developing the motion equations incorporating ecological uncertainties, we utilized continuous Bayesian updating to refine players' beliefs about unknown parameters. It was demonstrated that as players continuously receive signals, their beliefs gradually converge to the true values of these parameters, stabilizing over time. This convergence was supported by classical probability theorems, including Fubini's Theorem, Chebyshev’s Inequality, and the Borel-Cantelli Lemma, etc. Furthermore, we constructed Nash Equilibrium strategies for asymmetric players based on updated beliefs at any given time, which guided optimal control actions.

\section*{Acknowledgements} 
The work of Hongwei Gao is supported by National Natural Science Foundation of China (Grant No. 72171126). The work of Zhang Ye is supported by the Shenzhen Sci-Tech Fund (Grant No. RCJC20231211090030059) and the National Natural Science Foundation of China (Grant No. 12171036).

\bibliographystyle{plain}
\bibliography{sample}


\appendix
 \section{Proof of Theorem \ref{variance in motion}}
\label{appendix:proof2}

We will demonstrate equation (\ref{variance}) by dividing the analysis into two distinct parts. First, we establish that
\begin{equation*}\label{step1}
    \lim_{t \to \infty} \frac{B(t)}{t}
\end{equation*}
is finite. Subsequently, we show that
\begin{equation}\label{step2}
    \lim_{t \to \infty} \frac{\kappa(t)(\alpha(t)-1)}{t}
\end{equation}
is infinite. Based on these results, given that the ratio of a finite value to an infinite value approaches zero, the limit of (\ref{variance}) tends to zero as \(t \to \infty\).

\textbf{Step 1}:
Starting with equation (\ref{belief_equation}), we express \(B(t)\) as follows:
\begin{equation*}\label{bt}
    B(t) = \int_0^t \frac{(k_0+s)(X(s)-M(s))^2}{k_0+s+1}ds.
\end{equation*}
Our goal is to analyze each term individually. To facilitate this, we rewrite \( \frac{B(t)}{t} \) in terms of its quadratic components:
\begin{equation}\label{bt2}
    \frac{B(t)}{t} = \int_0^t \frac{(k_0+s)X^2(s)}{t(k_0+s+1)}ds - \int_0^t \frac{2(k_0+s)X(s)M(s)}{t(k_0+s+1)}ds + \int_0^t \frac{(k_0+s)M^2(s)}{t(k_0+s+1)}ds.
\end{equation}

\textbf{Step 1.1}:
We begin by analyzing the first integral in equation (\ref{bt2}). For any given \(t \geq 1\), we define a new random variable:
\[
Y_1(s) = \frac{(k_0+s)X^2(s)}{t(k_0+s+1)}.
\]

The expected value of \(Y_1(s)\) is:
\[
\mathbb{E}[Y_1(s)] = \frac{(k_0 + s)(\sigma^2 + \mu^2)}{t(k_0 + s + 1)}.
\]

To analyze the integral of \(\mathbb{E}[Y_1(s)]\) from 0 to \(t\), we compute:
\begin{equation*}
\begin{split}
    \int_0^t \mathbb{E}[Y_1(s)] \, ds &= (\sigma^2 + \mu^2) \int_0^t \frac{k_0 + s}{t (k_0 + s + 1)} \, ds\\
    &= (\sigma^2 + \mu^2)\left(1 - \frac{1}{t} \left[\log(k_0 + t + 1) - \log(k_0 + 1)\right]\right).
\end{split}
\end{equation*}

Applying  L'H\^opital's Rule \cite{LHospital1696} to the limit of the integral as \(t \to \infty\):
\[
\lim_{t \to \infty} \frac{1}{t} \log \frac{k_0 + 1}{k_0 + t + 1} = \lim_{t \to \infty}-\frac{1}{(k_0 + t + 1)} = 0.
\]

Therefore, we conclude that:
\begin{equation}\label{1.1.1}
\lim_{t \to \infty} \int_{s=0}^t \mathbb{E}[Y_1(s)] \, ds = \sigma^2 + \mu^2 < \infty.
\end{equation}

Next, we demonstrate that
\[
\lim_{t \to \infty} \int_0^t \mathbb{D}[Y_1(s)] \, ds
\]
approaches a finite value.

After computation,
\begin{equation*}
    \begin{split}
       & \lim_{t \to \infty} \int_{s=0}^t \mathbb{D}[Y_1(s)] \, ds = \lim_{t \to \infty} \int_{s=0}^t\left(\frac{(k_0 + s)}{t (k_0 + s + 1)}\right)^2 (2\sigma^4 + 4\mu^2\sigma^2)\, ds\\
        &= \lim_{t \to \infty} \frac{2\sigma^4 + 4\mu^2\sigma^2}{t^2} \int_{s=0}^t \left( 1 - \frac{2}{k_0 + s + 1} + \frac{1}{(k_0 + s + 1)^2}\right)\, ds\\
        &=\lim_{t \to \infty} \frac{2\sigma^4 + 4\mu^2\sigma^2}{t^2} \left( t - 2[\log(k_0 + t + 1) - \log(k_0 + 1)] + \frac{1}{k_0 + 1} - \frac{1}{k_0 + t + 1} \right)\\
        &=(2\sigma^4 + 4\mu^2\sigma^2)\lim_{t \to \infty}\left(\frac{1}{t}+\frac{2}{t^2} \log \frac{k_0 + 1}{k_0 + t + 1}+\frac{1}{t^2(k_0 + 1)}-\frac{1}{t^2(k_0 +t+ 1)}\right)\\
        &=0,
    \end{split}
\end{equation*}
it follows that: 
\begin{equation}\label{1.1.2}
     \lim_{t \to \infty} \int_{s=0}^t \mathbb{D}[Y_1(s)] \, ds=0<\infty.
\end{equation}

From conditions (\ref{1.1.1}) and (\ref{1.1.2}), we conclude, according to Proposition \ref{finite}, that the integral 
\[
\lim_{t \to \infty} \int_0^t Y_1(s) \, ds < \infty.
\]

\textbf{Step 1.2:} Next, we proceed to demonstrate that the integral of the second term in equation (\ref{bt2}) also approaches a finite value as \(t \to \infty\). According to Proposition \ref{finite}, it is sufficient to show that both the expectation and variance integrals converge. 

We define:
\begin{equation*}
    Y_2(s) = \frac{X(s) \left(\int_0^s X(\tau) \, d\tau + \mu_0 k_0\right)}{t (k_0 + s + 1)},
\end{equation*} 
where $X(s) \sim N(\mu, \sigma^2)$.

The expected value of \(Y_2(s)\) is:
\[
\mathbb{E}[Y_2(s)] = \frac{\mu (\mu s + \mu_0 k_0)}{t (k_0 + s + 1)}.
\]

Integrating \(\mathbb{E}[Y_2(s)]\) from 0 to \(t\):
\begin{equation}\label{23}
    \begin{split}
        &\int_0^t \mathbb{E}[Y_2(s)] \, ds = \int_0^t \frac{\mu (\mu s + \mu_0 k_0)}{t (k_0 + s + 1)} \, ds\\
        &=\frac{\mu^2}{t}\int_0^t\frac{s }{(k_0 + s + 1)} \, ds+ \frac{\mu\mu_0 k_0}{t}\int_0^t \frac{1 }{(k_0 + s + 1)} \, ds\\
        &=\frac{\mu^2}{t}\left(t - (k_0 + 1) \left(\log(k_0 + 1 + t) - \log(k_0 + 1)\right)\right) \\&+\frac{\mu\mu_0 k_0}{t}\left(\log(k_0 + 1 + t) - \log(k_0 + 1)\right).      
    \end{split}
\end{equation}

As $t \to \infty$, we evaluate the limit of (\ref{23}):
\[
\lim_{t \to \infty} \left[ \mu^2 \left(1 - \frac{(k_0+1)}{t} \log\frac{k_0 + 1 + t}{k_0 + 1}\right) + \mu \mu_0 k_0 \frac{1}{t} \log\frac{k_0 + 1 + t}{k_0 + 1} \right] = \mu^2<\infty.
\]

Therefore, 
\begin{equation*}\label{1.2.1}
     \lim_{t \to \infty} \int_{s=0}^t \mathbb{E}[Y_2(s)] \, ds=\mu^2<\infty.
\end{equation*}

The variance of \(Y_2(s)\) is given by:
\[
\mathbb{D}[Y_2(s)] = \mathbb{E}[(Y_2(s))^2] - \left(\mathbb{E}[Y_2(s)]\right)^2,
\]
which can be computed as:
\[
\mathbb{D}[Y_2(s)] = \frac{\left(\sigma^2 + \mu^2\right)\left(s \sigma^2 + \left(\mu s + \mu_0 k_0\right)^2\right)}{t^2 \left(k_0 + s + 1\right)^2} - \left(\frac{\mu \left(\mu s + \mu_0 k_0\right)}{t \left(k_0 + s + 1\right)}\right)^2.
\]

The integral of \(\mathbb{D}[Y_2(s)]\) over the interval from 0 to \(t\) is given by:
\[
\int_0^t \mathbb{D}[Y_2(s)] ds = \int_0^t \left( \frac{\sigma^2 \left(s \sigma^2 + \left(\mu s + \mu_0 k_0\right)^2\right)}{t^2 \left(k_0 + s + 1\right)^2} - \left(\frac{\mu \left(\mu s + \mu_0 k_0\right)}{t \left(k_0 + s + 1\right)}\right)^2 \right) ds.
\]

After evaluating the integral, we obtain:
\[
\int_0^t \mathbb{D}[Y_2(s)] ds = \sigma^2 \frac{-2 k_0 (1 + k_0) \mu \mu_0 + k_0^2 \mu_0^2 - (1 + k_0) \sigma^2 + (1 + k_0) \mu^2 (1 + 2 k_0 + t)}{t (1 + k_0) (1 + k_0 + t)} 
\]
\[
+ \frac{\sigma^2 (\mu (\mu + 2 k_0 \mu - 2 k_0 \mu_0) - \sigma^2)(\log(1 + k_0) - \log(1 + k_0 + t))}{t^2}.
\]

As \(t \to \infty\), the limit of the integral is:
\[
\lim_{t \to \infty} \int_0^t \mathbb{D}[Y_2(s)] ds = 0 < \infty.
\]

Thus, we conclude that:
\[
\lim_{t \to \infty} \int_0^t Y_2(s) ds < \infty.
\]

\textbf{Step 1.3}:  
We now proceed to prove the finiteness of the integral for the last part of equation (\ref{bt2}). In this context, we define a new random variable:
\[
Y_3(s) = \frac{(k_0+s)M^2(s)}{t(k_0+s+1)},
\]
where \( M(s) \) is given by:
\[
M(s) = \frac{1}{k_0 + s+1} \left(\int_0^s X(\tau) \, d\tau + \mu_0 k_0\right).
\]

The expected value of \( Y_3(s) \), \( \mathbb{E}[Y_3(s)] \), is calculated as follows:
\[
\mathbb{E}[Y_3(s)] = \frac{(k_0+s)}{t(k_0+s+1)} \left[\left(\frac{\mu s + \mu_0 k_0}{k_0 + s+1}\right)^2 + \left(\frac{1}{k_0 + s+1}\right)^2 \sigma^2 s\right].
\]

We can compute:
\begin{equation*}
\begin{split}
   & \int_0^t \mathbb{E}[Y_3(s)] \, ds = \frac{1}{t}\left(\mu^2 t - k_0^2 (\mu - \mu_0)^2 \log(k_0) + (\mu + k_0 \mu - k_0 \mu_0)^2 \log(1 + k_0)\right)\\
& \quad + \frac{1}{t} \left(k_0^2 (\mu - \mu_0)^2 \log(k_0 + t) - (\mu + k_0 \mu - k_0 \mu_0)^2 \log(1 + k_0 + t)\right)\\
& \quad + \frac{\sigma^2}{t} \left(k_0 \log(k_0) - (1 + k_0) \log(1 + k_0) - k_0 \log(k_0 + t) + (1 + k_0) \log(1 + k_0 + t)\right).
\end{split}
\end{equation*}

Thus, we obtain:
\[
\lim_{t \to \infty} \int_0^t \mathbb{E}[Y_3(s)] \, ds = \mu^2 < \infty.
\]

Next, we explore the convergence of the variance. Given that \( M(s) \) is a transformation of a Gaussian process, we use the following expressions:
\[
\mathbb{E}[M(s)] = \frac{\mu s + \mu_0 k_0}{k_0 + s+1},
\]
and
\[
\mathbb{D}[M(s)] = \left(\frac{1}{k_0 + s+1}\right)^2 \sigma^2 s.
\]
Thus, the second moment of \( M(s) \) is:
\[
\mathbb{E}[M^2(s)] = \left(\mathbb{E}[M(s)]\right)^2 + \mathbb{D}[M(s)] = \left(\frac{\mu s + \mu_0 k_0}{k_0 + s+1}\right)^2 + \left(\frac{\sigma \sqrt{s}}{k_0 + s+1}\right)^2.
\]

The fourth moment of a Gaussian random variable \( W \) with mean \( \mu \) and variance \( \sigma^2 \) is:
\[
\mathbb{E}[W^4] = 3\sigma^4 + 6\sigma^2 \mu^2 + \mu^4.
\]
Thus, applying this to \( M(s) \):
\[
\mathbb{E}[M^4(s)] = 3\left(\frac{\sigma^2 s}{(k_0 + s+1)^2}\right)^2 + 6\left(\frac{\sigma^2 s}{(k_0 + s+1)^2}\right)\left(\frac{\mu s + \mu_0 k_0}{k_0 + s+1}\right)^2 + \left(\frac{\mu s + \mu_0 k_0}{k_0 + s+1}\right)^4.
\]

Hence, the second moment of \( Y_3(s) \) is:
\[
\mathbb{E}[(Y_3(s))^2] = \frac{(k_0+s)^2}{t^2(k_0+s+1)^2} \mathbb{E}[M^4(s)].
\]

The variance of \( Y_3(s) \), \( \mathbb{D}[Y_3(s)] \), is given by:
\[
\mathbb{D}[Y_3(s)] = \mathbb{E}[(Y_3(s))^2] - \left(\mathbb{E}[Y_3(s)]\right)^2.
\]

After simplifying the expressions, we find:
\begin{equation*}
\begin{split}
    &\mathbb{D}[Y_3(s)] =  \frac{(k_0+s)^2}{t^2(k_0+s+1)^2} \left[ 3\left(\frac{\sigma^2 s}{(k_0 + s+1)^2}\right)^2 + 6\left(\frac{\sigma^2 s}{(k_0 + s+1)^2}\right)\left(\frac{\mu s + \mu_0 k_0}{k_0 + s+1}\right)^2 \right.\\
    &\left. + \left(\frac{\mu s + \mu_0 k_0}{k_0 + s+1}\right)^4 \right] - \left( \frac{(k_0+s)}{t(k_0+s+1)} \left[\left(\frac{\mu s + \mu_0 k_0}{k_0 + s+1}\right)^2 + \left(\frac{\sigma \sqrt{s}}{k_0 + s+1}\right)^2\right] \right)^2.
\end{split}
\end{equation*}

Finally, after calculation, we get:
\begin{equation}\label{inte_var}
    \lim_{t \to \infty} \int_0^t \mathbb{D}[Y_3(s)] \, ds = 0 < \infty.
\end{equation}

Due to the complexity of the calculation, we omit the explicit form of the result. However, it is clear that the denominator of \( \mathbb{D}[Y_3(s)] \) is quadratic in time, while the numerator grows linearly with time. As time progresses, the limit in equation (\ref{inte_var}) approaches zero.

Thus, we conclude that:
\[
\lim_{t \to \infty} \int_0^t Y_3(s) \, ds = 0 < \infty.
\]

\textbf{Step 2:} 
Finally, we demonstrate that the limit in equation (\ref{step2}) is infinite. Consider the expression:
\begin{equation*}
    \begin{split}
        \lim_{t \to \infty} \frac{(k_0 + t)(\alpha_0 + \frac{1}{2} t - 1)}{t} &= \lim_{t \to \infty} \frac{k_0 \alpha_0 + k_0\left(\frac{1}{2} t - 1\right) + t \alpha_0 + \frac{1}{2} t^2 - t}{t} \\
        &= \infty.
    \end{split}
\end{equation*}

Thus, combining these results, we conclude that as \(t \to \infty\), the expression for (\ref{variance}) tends to zero, implying that our estimates stabilize over time.

\section{Proof of Proposition \ref{optimal control}}
\label{appendix:proof3}

 The value function of player $i$ starting at time $t\in [0,\infty)$ defined by (\ref{motion_continuous}) and (\ref{expected_payoff_continuous}) is given by
\begin{equation*}
\begin{split}
   &V_i(t, S, \tau_i) =\\
   &\max _{u_i(k,\cdot)}\int_{k=t}^\infty e^{-\rho (k-t)}\left [u_{i}(k, \overline S(k); \tau_i, \cdot)\left(a_i-u_{i}-\sum_{j\not = i}^{n} {u_{j}^*(k, \overline S(k); \overline \tau_j(k), \cdot)}\right)-{\tau_i}\overline S(k) \right ] dt,\\
   & \text{where} \quad \overline S(t)= S(t).
\end{split}
\end{equation*}

According to the Hamilton-Jacobi-Bellman (HJB) equations:
\begin{equation*}\label{value_2}
\begin{split}
    & \rho V_i(k, \overline S(k), \tau_i) =\\
     &\max_{u_i(k, \cdot)} \Bigg\{ \left[ u_i(k,\overline S(k);\tau_i, \cdot) \left( a_i - u_i(k,\overline S(k); \tau_i, \cdot) - \sum_{j \neq i} {u_{j}^*(k, \overline S(k); \overline \tau_j(t), \cdot)} \right) - \tau_i \overline S \right] \\
     & + \frac{\partial V_i}{\partial S}(k, \overline S(k), \tau_i) \overline x(t) \left( u_i(k,\overline S(k);\tau_i,\cdot)+\sum_{j\neq i}^n{u_{j}^*(k, \overline S(k); \overline \tau_j(t), \cdot)}\right) \\& -\frac{\partial V_i}{\partial S}(k, \overline S(k), \tau_i)\left(1 - \overline x(t)\delta\right)  \overline S\Bigg\}.
\end{split}
\end{equation*}

Assume that the value function is expressed in the form:
\begin{equation*}
    V_i(k, \overline S(k), \tau_i)=A_i(\tau_i)\overline S+B_i.
\end{equation*}

Substituting the conjectured value function into the HJB equation (\ref{value_22}), we obtain:
\begin{equation}\label{value_22}
    \begin{split}
        &\rho (A_i(\tau_i)\overline S+B_i) =\max_{u_i(t, \cdot)} \Bigg\{ \left[ u_i(k, \overline S;\tau_i, \cdot) \left( a_i - u_i - \sum_{j \neq i} {u_{j}^*(k, \overline S(k); \overline \tau_j(t), \cdot)} \right) - \tau_i \overline S \right] \\
     & + A_i(\tau_i) \left(\overline x(t) \left( u_i(k,\overline S;\tau_i,\cdot)+\sum_{j\neq i}^n{u_{j}^*(k, \overline S(k); \overline \tau_j(t), \cdot)}\right) -\left(1 - \overline x(t)\delta\right)  \overline S\right)\Bigg\}.
    \end{split}
\end{equation}

Suppose that \( u_i^*(k,\overline S; \tau_i,\cdot) \) is a linear function of \( \tau_i \), so we can simplify it as:
\begin{equation}\label{simplify_u_i}
    u_i^*=f_{1}^i(k)+f_2^i(k)\tau_i.
\end{equation}

Thus, we can write the value function as:
\begin{equation}\label{value_3}
    \begin{split}
        &\rho (A_i(\tau_i)\overline S+B_i) \\&= \max \limits_{u_i(k, \cdot)} \Bigg\{ u_{i}(k, \overline S(k); \tau_i, \cdot)\left(a_i-u_{i}(k, \overline S(k); \tau_i, \cdot)-\sum_{j\not = i}^{n} \left(f_1^j(k)+f_2^j(k)\overline\tau_j\right)\right)-{\tau_i}\overline S \\
        &+ A_i(\tau_i)\cdot\left(\overline x(t) u_i(t,\overline S;\tau_i, \cdot)+ \overline x(t) \sum_{j=1,\atop j\neq i}^n \left(f_1^j(k)+f_2^j(k)\overline\tau_j\right)-(1- \overline x(t)\delta)   \overline S\right)\Bigg\}.
    \end{split}
\end{equation}

The maximization problem in (\ref{value_3}) yields the following strategy for player \( i \):
\begin{equation}\label{generalize}
     u_i^*(k, \overline S(k); \tau_i, \{\overline\tau_j\}_{j \in N}, \overline x(t))=\frac{a_i-\sum_{ j\neq i}^n\left(f_1^j(k)+f_2^j(k)\overline\tau_j\right)+A_i(\tau_i)\overline x(t)}{2}.
\end{equation}

Substituting (\ref{generalize}) into (\ref{value_3}), we obtain the following differential equation for \( A_i(\tau_i) \):
\begin{equation*}
      \rho A_i(\tau_i)=\tau_i+A_i(\tau_i)(1- \overline x(t)\delta).
\end{equation*}

The solution is given by:
\begin{equation*}
    A_i(\tau_i)=-\frac{\tau_i}{1- \overline x(t)\delta-\rho}.
\end{equation*}

Therefore, the Nash equilibrium with continuous Bayesian updating of player $i\in N$ is:
\begin{equation}\label{simplify_u_i2}
    u_i^* = \frac{a_i}{2}-\frac{\sum_{j \neq i}^n \left( f_1^j(k) + f_2^j(k)\overline\tau_j \right)}{2}-\frac{ \tau_i\overline x(t) }{2(1- \overline x(t)\delta-\rho)}.
\end{equation}

By comparing (\ref{simplify_u_i}) and (\ref{simplify_u_i2}), we can derive the specific formula for \( f_2^i(k) \) from the coefficient of \( \tau_i \) as follows:
\begin{equation}\label{f2}
    f_2^i(k)=-\frac{\overline x(t)}{2(1- \overline x(t)\delta-\rho)}\stackrel{\text { def }}{=}f_2(k), \quad i\in N.
\end{equation}

For \( f_1^i(k) \), we have:
\begin{equation}\label{f1}
    f_1^i(k)= \frac{ a_i - \sum_{j \neq i}^n \left( f_1^j(k) + f_2^j(k) \overline\tau_j \right)}{2}.
\end{equation}

Subtracting \( \frac{f_1^i(k)}{2} \) from (\ref{f1}), we obtain:
\begin{equation}\label{f11}
    f_1^i(k)=a_i-\sum_{j=1}^n f_1^j(k)-\sum_{ j\neq i}^n f_2(k)\overline\tau_j.
\end{equation}

Summing over \( i=1,2,...,n\), we get:
\begin{equation*}
    \sum_{i=1}^n{f_1^i(k)}={a}-n{\sum_{j=1}^n f_1^j(k)}-{f_2(k)(n-1)\sum_{j=1}^n\overline\tau_j},
\end{equation*}
where \( a= \sum_{i=1}^n a_i \).

Thus, we find:
\begin{equation}\label{sum}
    \sum_{i=1}^n f_1^i(k)=\frac{a}{n+1}+\frac{(n-1)\overline x(t)\sum_{j=1}^n\overline\tau_j}{2(n+1)(1- \overline x(t)\delta)}.
\end{equation}

Substituting (\ref{sum}) into (\ref{f11}), we obtain:
\begin{equation*}\label{f1final}
f_1^i(k)=a_i-\left(\frac{a}{n+1}+\frac{(n-1)\overline x(t)\sum_{j=1}^n\overline\tau_j}{2(n+1)(1- \overline x(t)\delta-\rho)}\right)+\frac{\overline x(t)}{2(1- \overline x(t)\delta-\rho)}\sum_{ j\neq i}^n \overline\tau_j.
\end{equation*}

Let \( \overline c(t) \) denote \( -\frac{\overline x(t)}{1- \overline x(t)\delta-\rho} \), then we define:
\begin{gather*}
     f_2^i(k)=\frac{ \overline c(t)}{2},\\
    f_1^i(k)=a_i-\frac{a}{n+1}+\frac{ \overline c(t)(n-1)\sum_{j=1}^n \overline\tau_j}{2(n+1)}-\frac{ \overline c(t)}{2}\sum_{j\neq i}^n \overline\tau_j.
\end{gather*}

Since
\[
\sum_{j \neq i}^n \overline{\tau}_j = \sum_{j=1}^n \overline{\tau}_j - \overline{\tau}_i,
\]
 we can rewrite \( f_1^i(k) \) and substitute it into (\ref{simplify_u_i}), and then let $k=t$ to obtain the final form of the Nash equilibrium with continuous Bayesian updating.

\end{document}